\documentclass[english]{article}

\usepackage[T1]{fontenc}
\usepackage[latin9]{inputenc}
\usepackage{amsthm}
\usepackage{listings}
\usepackage{kbordermatrix}
\usepackage{graphicx}
\usepackage{caption}
\usepackage{subcaption}
\usepackage{epstopdf}
\usepackage{color} 

\definecolor{mygreen}{RGB}{28,172,0} 
\definecolor{mylilas}{RGB}{170,55,241}

\usepackage{babel,amsmath,amsthm,amsfonts,amssymb,enumerate}

\newtheorem{lemma}{Lemma}
\newtheorem{definition}{Definition}

\newtheorem{theorem}{Theorem}

\newcommand{\X}{\mathcal{X}}
\newcommand{\eps}{\varepsilon}
\newcommand{\1}{\textbf{1}}
\newcommand{\G}{\mathcal{G}}
\newcommand{\Gp}{\mathcal{G}_+}
\newcommand{\Gm}{\mathcal{G}_-}
\newcommand{\Dp}{D_{\mathcal{G}}^+}
\newcommand{\Dm}{D_{\mathcal{G}}^-}
\newcommand{\pp}{\frac{\alpha\log(n)}{n}}
\newcommand{\qq}{\frac{\beta\log(n)}{n}}
\newcommand{\A}{\frac{\alpha \log(n)}{n}}
\newcommand{\B}{\frac{\beta \log(n)}{n}}

\newcommand{\ap}{\alpha^+_{ij}}
\newcommand{\am}{\alpha^-_{ij}}
\newcommand{\E}{\mathbb{E}}

\newcommand{\OOO}{O}
\newcommand{\p}{\mathbb{P}}

\newcommand{\tr}{\operatorname{Tr}}

\begin{document}

\title{Exact Recovery in the Stochastic Block Model}

\author{Emmanuel Abbe~\thanks{Program in Applied and Computational Mathematics (PACM) and Department of Electrical Engineering, Princeton University, Princeton, NJ 08544, USA ({\tt eabbe@princeton.edu}).}
\and
Afonso S.~Bandeira~\thanks{PACM, Princeton University, Princeton, NJ 08544, USA ({\tt ajsb@math.princeton.edu}). ASB was supported by AFOSR Grant No. FA9550-12-1-0317.}
\and
Georgina Hall~\thanks{Department of Operations Research and Financial Engineering, Princeton University, Princeton, NJ 08544, USA ({\tt gh4@princeton.edu}).}
}

%
%

\date{}

\maketitle

\begin{abstract}
The stochastic block model (SBM) with two communities, or equivalently the planted bisection model, is a popular model of random graph exhibiting a cluster behaviour. 
In the symmetric case, the graph has two equally sized clusters and vertices connect with probability $p$ within clusters and $q$ across clusters. 
 In the past two decades, a large body of literature in statistics and computer science has focused on providing lower-bounds on the scaling of $|p-q|$ to ensure exact recovery. 
In this paper, we identify a sharp threshold phenomenon for exact recovery: if $\alpha=pn/\log(n)$ and $\beta=qn/\log(n)$ are constant (with $\alpha>\beta$), recovering the communities with high probability is possible if $\frac{\alpha+\beta}{2} - \sqrt{\alpha \beta}>1$ and impossible if $\frac{\alpha+\beta}{2} - \sqrt{\alpha \beta}<1$. In particular, this improves the existing bounds. 
This also sets a new line of sight for efficient clustering algorithms. While maximum likelihood (ML) achieves the optimal threshold (by definition), it is in the worst-case NP-hard. This paper proposes an efficient algorithm based on a semidefinite programming relaxation of ML, which is proved to succeed in recovering the communities close to the threshold, 
 while numerical experiments suggest it may achieve the threshold. An efficient algorithm which succeeds all the way down to the threshold is also obtained using a partial recovery algorithm combined with a local improvement procedure. 
\end{abstract}

\section{Introduction}
%
%

Learning community structures in graphs is a central problem in machine learning, computer science and complex networks.   
Increasingly, data is available about interactions among agents (e.g., social, biological, computer or image networks), and the goal is to infer from these interactions communities that are alike or complementary.
As the study of community detection grows at the intersections of various fields, in particular computer science, machine learning, statistics and social computing, the notions of clusters, the figure of merits and the models vary significantly, often based on heuristics (see \cite{airoldi} for a survey). As a result, the comparison and validation of clustering algorithms remains a major challenge. Key enablers  to benchmark algorithms and to measure the accuracy of clustering methods are statistical network models. More specifically, the stochastic block model has been at the center of the attention in a large body of literature  \cite{holland,sbm1,sbm3,sbm4,sbm5,dyer,newman2,mcsherry,sbm-book,decelle,Massoulie_SBM,Mossel_SBM1,Mossel_SBM2}, as a testbed for algorithms (see \cite{sbm-algos} for a survey) as well as a scalable model for large data sets (see \cite{prem} and reference therein). On the other hand, the fundamental analysis of the stochastic block model (SBM) is still holding major open problems, as discussed next. 




The SBM can be seen as an extension of the Erd\H{o}sR\'enyi (ER) model \cite{ER-seminal,ER2}. 
In the ER model, edges are placed independently with probability $p$, providing a models described by a single parameter.
This model has been (and still is) a source of intense research activity, in particular due to its phase transition phenomena.  
It is however well known to be too simplistic to model real networks, in particular due to its strong homogeneity and absence of community structure. 
The stochastic block model is based on the assumption that agents in a network connect not independently but based on their profiles, or equivalently, on their community assignment. 
More specifically, each node $v$ in the graph is assigned a label $x_v \in \X$, where $\X$ denotes the set of community labels, and each pair of nodes $u,v \in V$ is connected with probability $p(x_u,x_v)$, where $p(\cdot, \cdot)$ is a fixed probability matrix. Upon observing the graph (without labels), the goal of community detection is to reconstruct the community assignments, with either full or partial recovery.

Of particular interest is the SBM with two communities and symmetric parameters, also known as the planted bisection model, denoted in this paper by $\mathcal{G}(n,p,q)$, with $n$ an even integer denoting the number of vertices. In this model, the graph has two clusters of equal size, and the probabilities of connecting are $p$ within the clusters and $q$ across the clusters (see Figure \ref{cluster-plot}). Of course, one can only hope to recover the communities up to a global flip of the labels, in other words, only the partition can be recovered. Hence we use the terminology {\it exact recovery} or simply {\it recovery} when the partition is recovered correctly with high probability (w.h.p.), i.e., with probability tending to one as $n$ tends to infinity. When $p=q$, it is clearly impossible to recover the communities, whereas for $p>q$ or $p<q$, one may hope to succeed in certain regimes.
While this is a toy model, it captures some of the central challenges for community detection. 

\begin{figure}[h]
\centering
\begin{subfigure}{.5\textwidth}
  \centering
  \includegraphics[width=.5\linewidth]{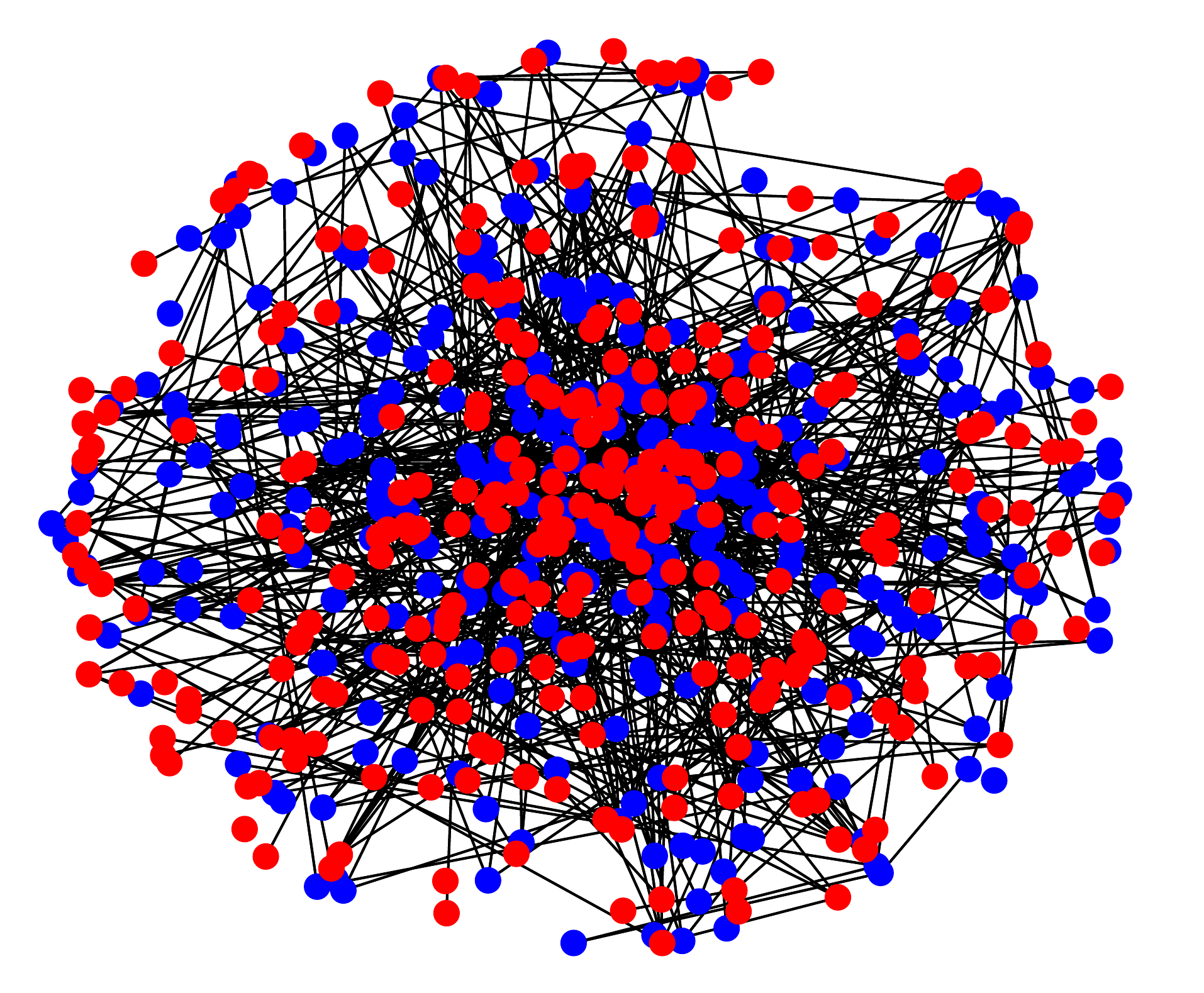}
\end{subfigure}%
\begin{subfigure}{.5\textwidth}
  \centering
  \includegraphics[width=.9\linewidth]{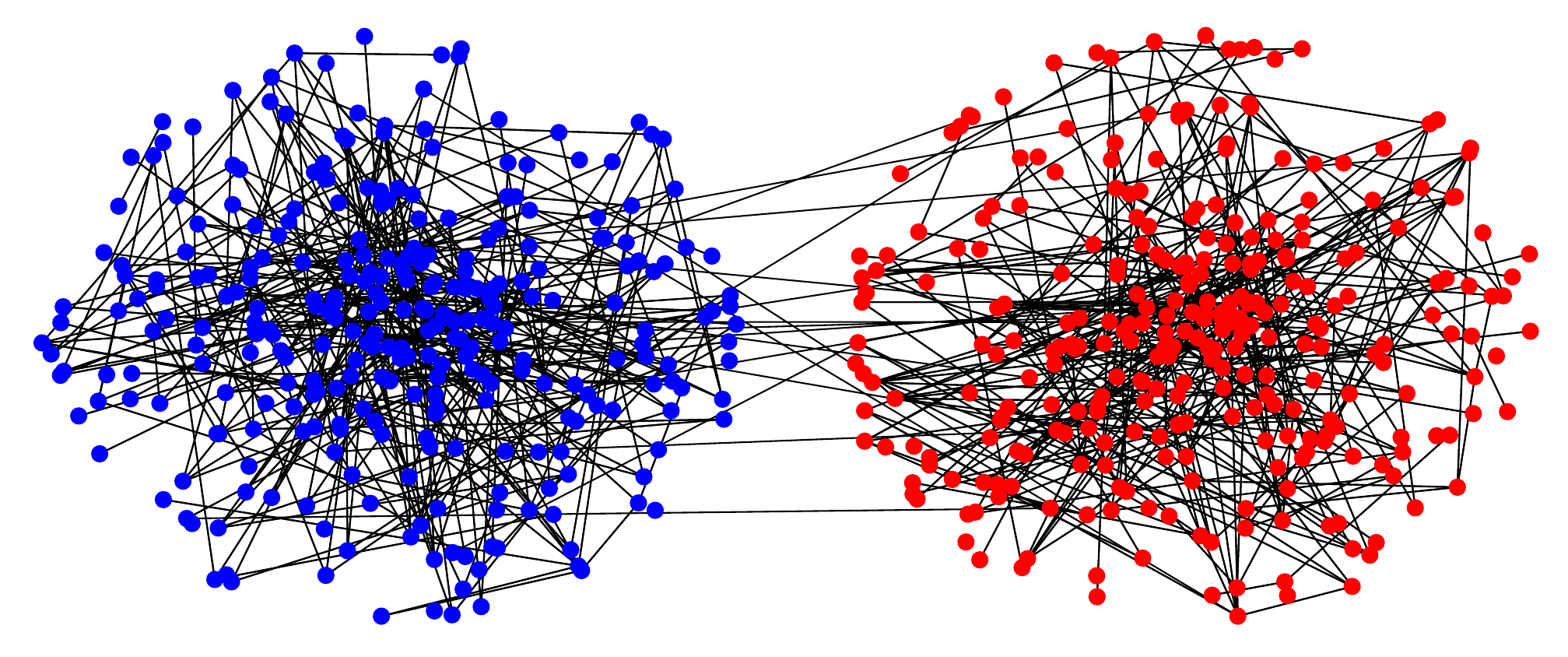}
\end{subfigure}
\caption{A graph generated form the stochastic block model with 600 nodes and 2 communities, scrambled on the left and clustered on the right. Nodes in this graph connect with probability $p=6/600$ within communities and $q=0.1/600$ across communities. 
}
\label{cluster-plot}
\end{figure}

A large body of literature in statistics and computer science \cite{bui,dyer,boppana,snij,jerrum,condon,carson,mcsherry,bickel,rohe,choi} has focused on determining lower-bounds on the scaling of $|p-q|$ for which efficient algorithms succeed in recovering the two communities in $\mathcal{G}(n,p,q)$.  
We overview these results in the next section. The best bound seems to come from \cite{mcsherry}, ensuring recovery for $(p - q)/\sqrt{p} \geq \Omega (\sqrt{\log(n)/n })$, and has not be improved for more than a decade. More recently, a new phenomena has been identified for the SBM in a regime where $p=a/n$ and $q=b/n$ \cite{decelle}. In this regime, exact recovery is not possible, since the graph is, with high probability, not connected.
However, partial recovery is possible, and the focus has been shifted on determining for which regime of $a$ and $b$ it is possible to obtain a reconstruction of the communities which is asymptotically better than a random guess (which gets roughly $50\%$ of accuracy). In other words, to recover only a proportion $1/2 + \eps$ of the vertices correctly, for some $\eps >0$. We refer to this reconstruction requirement as {\it detection}. In \cite{decelle}, it was conjectured that detection is possible if and only if $(a-b)^2> 2 (a+b)$. This is a particularly fascinating and strong conjecture, as it provides a necessary and sufficient condition for detection with a sharp closed-form expression. The study of this regime was initiated with the work of Coja-Oghlan \cite{coja-sbm}, which obtains detection when $(a-b)^2> 2 \log(a+b) (a+b)$ using spectral clustering on a trimmed adjacency matrix. 
The conjecture was recently proved by Massoulie \cite{Massoulie_SBM} and Mossel et al.\ \cite{Mossel_SBM2} using two different efficient algorithms. The impossibility result was first proved in \cite{Mossel_SBM1}.

While the sparse regime with constant degree points out a fascinating threshold phenomena for the detection property, it  also raises a natural question: does exact recovery also admit a similar phase transition? Most of the literature has been focusing on the scaling of the lower-bounds, often up to poly-logarithmic terms, and the answer to this question appears to be currently missing in the literature. In particular, we did not find tight impossibility results, or guarantees of optimality of the proposed algorithms. This paper answers this question, establishing a sharp phase transition for recovery, obtaining a tight bound with an efficient algorithm achieving it. 



\section{Related works}\label{related}

There has been a significant body of literature on the recovery property for the stochastic block model with two communities $\mathcal{G}(n,p,q)$, ranging from computer science and statistics literature to machine learning literature. We provide next a partial\footnote{The approach of McSherry was recently simplified and extended in \cite{Vu-arxiv}.} list of works that obtain bounds on the connectivity parameters to ensure recovery with various algorithms:
\begin{center}
  \begin{tabular}{| l |  c |c   | }
  \hline
\cite{bui} Bui, Chaudhuri, Leighton, Sipser '84  & min-cut method& $p = \Omega(1/n), q=o(n^{-1-4/((p+q)n)})$\\
  \hline
\cite{dyer} Dyer, Frieze '89  & min-cut via degrees & $p -q = \Omega(1)$\\
  \hline
\cite{boppana} Boppana '87  & spectral method & $(p -q)/\sqrt{p+q} = \Omega(\sqrt{\log(n)/n})$\\
  \hline
\cite{snij} Snijders, Nowicki '97  & EM algorithm & $p -q = \Omega(1)$\\
\hline
\cite{jerrum} Jerrum, Sorkin '98  &Metropolis aglorithm & $p -q= \Omega(n^{-1/6+\eps})$\\
   \hline
\cite{condon} Condon, Karp '99 & augmentation algorithm & $p -q= \Omega(n^{-1/2+\eps})$\\
  \hline
\cite{carson} Carson, Impagliazzo '01  & hill-climbing algorithm & $p-q= \Omega(n^{-1/2} \log^4(n))$\\
 \hline
\cite{mcsherry} Mcsherry '01  & spectral method & $(p - q)/\sqrt{p} \geq \Omega (\sqrt{\log(n)/n })$\\
 \hline
\cite{bickel} Bickel, Chen '09  & N-G modularity & $(p -q)/\sqrt{p+q} = \Omega(\log(n)/\sqrt{n})$\\
 \hline
\cite{rohe} Rohe, Chatterjee, Yu '11  & spectral method & $p -q= \Omega(1)$\\
  \hline
  \end{tabular}
\end{center}
While these algorithmic developments are impressive, we next argue how they do not reveal the sharp behavioral transition that takes place in this model. In particular, we will obtain an improved bound that is shown to be tight. 


\section{Information theoretic perspective and main results}
In this paper, rather than starting with a specific algorithmic approach,
we first seek to establish the information-theoretic threshold for recovery irrespective of efficiency requirements.
Obtaining an information-theoretic benchmark, we then seek for an efficient algorithm that achieves it.   
There are several reasons to expect that an information-theoretic phase transition takes place for recovery in the SBM:
\begin{itemize}
\item From a random graph perspective, note that recovery requires the graph to be at least connected (with high probability), hence $(\alpha+\beta)/2>1$, for $p=\alpha \log(n)/n$ and $q=\beta \log(n)/n$ is necessary. In turn, if $\alpha=0$ or $\beta=0$, then $(\alpha+\beta)/2<1$ prohibits recovery (since the model has either two separate Erd\H{o}s-Re\'enyi graphs that are not connected, or a bipartite Erd\H{o}s-Re\'enyi graph which is not connected). So one can expect that recovery take place in the regime $p=\alpha \log(n)/n$ and $q=\beta \log(n)/n$, if and only if $f(\alpha,\beta)>1$ for some function $f$ that satisfies $f(\alpha,0)=\alpha/2$, $f(0,\beta)=\beta/2$ and where $f(\alpha,\beta)>1$ implies $(\alpha+\beta)/2>1$. In particular, such a result has been shown to take place for the detection property  \cite{Massoulie_SBM,Mossel_SBM2}, where a giant component is necessary, i.e., $(a+b)/2 >1$ for $p=a/n$ and $q=b/n$, and where detection is shown to be possible if and only if $(a+b)/2 + 2ab/(a+b)>1$ (which is equivalent to $(a-b)^2 > 2(a+b)$).  
Note also that the regime $p=\alpha \log(n)/n$ and $q=\beta \log(n)/n$ is the bottleneck regime for recovery, as other regimes lead to extremal behaviour of the model (either trivially possible or impossible to recover the communities).  
 
\item From an information theory perspective, note that the SBM can be seen as specific code on a discrete memoryless channel. Namely, the community assignment is a vector $x \in \{0,1\}^n$,  the graph is a vector (or matrix) $y \in \{0,1\}^N$, $N={n \choose 2}$, where $y_{ij}$ is the output of $x_i \oplus x_j$ through the discrete memoryless channel $(\begin{smallmatrix} 1-p & p \\ 1-q & q \end{smallmatrix})$, for $1 \leq i<j \leq n$. The problem is hence to decode $x^n$ from $y^N$ correctly with high probability. 

This information theory model is a specific structured channel: first the channel is memoryless but it is not time homogeneous, since $p=a \log(n)/n$ and $q=b \log(n)/n$ are scaling with $n$. 
Then the code has a specific structure, it has constant right-degree of 2 and constant left-degree of $n-1$, and rate $2/(n-1)$. However, as shown in \cite{abbetoc} for the constant-degree regime, this model can be approximated by another model where the sparsity of the channel (i.e., the fact that $p$ and $q$ tend to 0) can be transferred to the code, which becomes an LDGM code of constant degree 2, and for which maximum-likelihood is expected to have a phase transition \cite{abbetoc,kumar}. It is then legitimate to expect a phase transition, as in coding theory, for the recovery of the input (the community assignment) from the output (the graph). 
\end{itemize}

\begin{figure}[h]
\begin{center}
\includegraphics[scale=.33]{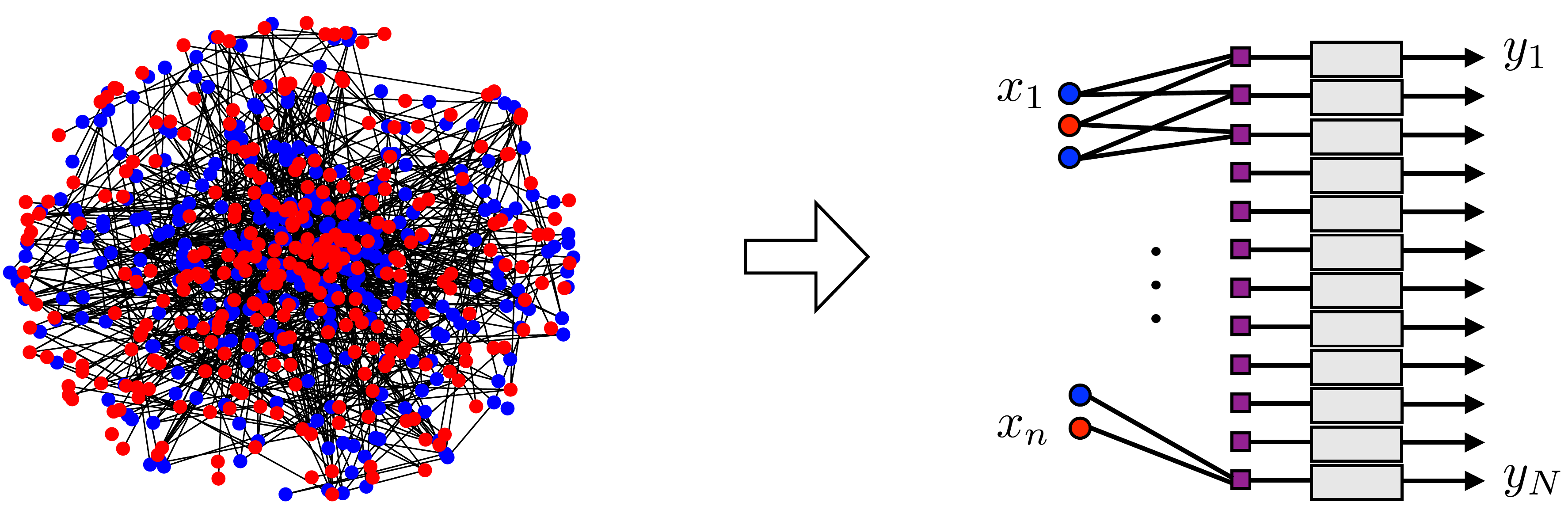}
\caption{Figure 2: A graph model like the stochastic block model where edges are drawn dependent on the node profiles (e.g., binary profiles) can be seen as a special (LDGM) code on a memoryless channel.}
\label{meet}
\end{center}
\end{figure}

To establish the information-theoretic limit, note that, as for channel coding, the algorithm maximizing the probability of reconstructing the communities correctly is the Maximum A Posteriori (MAP) decoding. Since the community assignment is uniform, MAP is in particular equivalent to Maximum Likelihood (ML) decoding. Hence if ML fails in reconstructing the communities with high probability when $n$ diverges, there is no algorithm (efficient or not) which can succeed with high probability. However, ML amounts to finding a balanced cut (a bisection) of the graph which minimizes the number of edges across the cut (in the case $a>b$), i.e., the min-bisection problem, which is well-known to be NP-hard. Hence ML can be used\footnote{ML was also used for the SBM in \cite{choi}, requiring however poly-logarithmic degrees for the nodes.}  to establish the fundamental limit but does not provide an efficient algorithm, which we consider in a second stage.

We now summarize the main results of this paper. Theorem~\ref{theorem:mainlowerbound} and Theorem~\ref{theorem:main_upperbound_2} provide the information-theoretic limit for recovery. 
Theorem~\ref{theorem:mainlowerbound} establishes the converse, showing that the maximum likelihood estimator does not coincide with the planted partition w.h.p.\  if $(\alpha + \beta)/2 - \sqrt{\alpha \beta} < 1$ and Theorem~\ref{theorem:main_upperbound_2} states that ML succeeds w.h.p.\  if $(\alpha + \beta)/2 - \sqrt{\alpha \beta} > 1$. 
One can express the recovery requirement as 
\begin{align}
(\alpha + \beta)/2  > 1 + \sqrt{\alpha \beta}
\end{align}
where $(\alpha + \beta)/2  > 1$ is the requirement for the connectivity threshold (which is necessary), and the oversampling term $\sqrt{\alpha \beta}$ is needed to allow for recovery (this is also equivalent to $\sqrt{\alpha}- \sqrt{\beta} >1$ for $\alpha >\beta$). Analyzing ML requires a sharp analysis of the tail event of the sum of discrete random variables tending to constants with the number of summands. Interestingly, standard estimates \`a la CLT, Chernoff, or Sanov's Theorem do not provide the right answer in our regime due to the slow concentration taking place.

Note that the best bounds from the table of Section \ref{related} are obtained from \cite{boppana} and \cite{mcsherry}, which allow for recovery in the regime where $p=\alpha \log(n)/n$ and $q=\beta \log(n)/n$, obtaining the conditions $(\alpha-\beta)^2 > 64(\alpha+\beta)$ in~\cite{mcsherry} and $(\alpha-\beta)^2 > 72(\alpha+\beta)$ in~\cite{boppana}.  Hence, although these works reach the scaling for $n$ where the threshold takes place, they do not obtain the right threshold behaviour in terms the parameters $\alpha$ and $\beta$.

For efficient algorithms, we propose first an algorithm based on a semidefinite programming relaxation of ML, and show in Theorem~\ref{theorem:main_SDP_provenintheappendix} that it succeeds in recovering the communities w.h.p.\ when $(\alpha - \beta)^2 > 8(\alpha + \beta)+ 8/3(\alpha - \beta)$. This is shown by building a candidate dual certificate and showing that it indeed satisfies all the require properties, using Berstein's matrix inequality.
To compare this expression with the optimal threshold, the latter can be rewritten as $(\alpha - \beta)^2 > 4(\alpha + \beta) - 4$ and $\alpha + \beta >2$. 
The SDP is hence provably successful with a slightly looser threshold. It however already improves on the state of the art for exact recovery in the SBM, since the above condition is implied by $(\alpha-\beta)^2 > 31/3(\alpha+\beta)$, which improves on~\cite{mcsherry}. 
Moreover, numerical simulations suggest that the SDP algorithm works all the way down to the optimal threshold, and the analysis may not be tight. The success of the SDP algorithm under the model of this paper, suggest that it may have robustness properties relevant in practical contexts.     

Finally, we provide in Section \ref{black} an efficient algorithm whose guarantees match the information theoretical threshold, using an efficient partial recovery algorithm, followed by a procedure of local improvements. 

Summary of the regimes and thresholds:
\begin{center}
  \begin{tabular}{| l ||  c | c  | }
   \hline
    & Giant component  & Connectivity \\ 
   ER model $G(n,p)$ &  $p=\frac{c}{n}$, $c > 1$  & $p=\frac{c \log(n)}{n}$, $c>1$ \\ 
   &\cite{ER2} & \cite{ER2}\\
    \hline
    \hline
           & Detection  & Recovery \\ 
    SBM model $G(n,p,q)$ & $p=\frac{a}{n}$, $q=\frac{b}{n}$,  $(a-b)^2 >2(a+b)$  & $p=\frac{a \log(n)}{n}$, $q=\frac{b\log(n)}{n}$,  $\frac{a+b}{2} - \sqrt{ab} >1$ \\ 
    &\cite{Massoulie_SBM,Mossel_SBM2} & (This paper)\\
    \hline
  \end{tabular}
\end{center}

\section{Additional related literature}\label{lit}



From an algorithmic point of view, the censored block model investigated in \cite{Abbe_Z2SynchJournal,Abbe_Z2SynchER} is also related to this paper. It considers the following problem: $G$ is a random graph from the $ER(n,p)$ ensemble, and each node $v$ is assigned an unknown binary label $x_v$. For each edge $(i,j)$ in $G$, the variable $Y_{ij}=x_i+x_j+Z_{ij} \mod 2$ is observed, where $Z_{ij}$ are i.i.d.\ Bernouilli($\epsilon$) variables. The goal is to recover the values of the node variables from the $\{Y_{ij}\}$ variables. Matching bounds are obtained in \cite{Abbe_Z2SynchJournal,Abbe_Z2SynchER} for $\eps$ close to $1/2$, with an efficient algorithm based on SDP, which is related to the algorithm developed in this paper. 

Shortly after the posting of this paper on arXiv, a paper of Mossel, Neeman and Sly \cite{mossel-consist}, fruit of a parallel research effort, was posted for the recovery problem in $\mathcal{G}(n,p,q)$. In \cite{mossel-consist}, the authors obtained a similar type of result as in this paper, slightly more general, allowing in particular for the parameters $a$ and $b$ to depend on $n$ as long as both parameters are $\Theta(1)$. 

%
%

\section{Information theoretic lower bound}
In this section we prove an information theoretic lower bound for exact recovery on the stochastic block model. The techniques are similar to the estimates for decoding a codeword on a memoryless channel with a specific structured codes. 

Recall the $\mathcal{G}(n,p,q)$ stochastic block model: $n$ denotes the number of vertices in the graph, assumed to be even for simplicity, for each vertex $v \in [n]$, a binary label $X_v$ is attached, where $\{X_v\}_{v \in [n]}$ are uniformly drawn such that $|\{v \in [n] : X_v =1\}|=n/2$, and for each pair of distinct nodes $u,v \in [n]$, an edge is placed with probability $p$ if $X_u=X_v$ and $q$ if $X_u \neq X_v$, where edges are placed independently conditionally on the vertex labels. In the sequel, we consider $p=\alpha \log(n)/n$ and $q=\beta \log(n)/n$, and focus on the case $\alpha > \beta$ to simplify the writing.

\begin{theorem}\label{theorem:mainlowerbound}
Let $\alpha>\beta \geq 0$. If $(\alpha + \beta)/2 - \sqrt{\alpha \beta} < 1$, or equivalently, if either $\alpha + \beta < 2$ or $(\alpha-\beta)^2 < 4(\alpha+\beta) - 4$ and $\alpha + \beta \geq 2$, then ML fails in recovering the communities with probability bounded away from zero.
\end{theorem}

If $\beta = 0$, recovery is possibly if and only if there are no isolated nodes which is known to have a sharp threshold at $\alpha=2$. We will focus on $\alpha>\beta>0$.

Let $A$ and $B$ denote the two communities, each with $\frac{n}2$ nodes. 

Let
$$\gamma(n) = \log^3 n, \ \delta(n) = \frac{\log n}{\log \log n},$$
and let $H$ be a fixed subset of $A$ of size $\frac{n}{\gamma(n)}$.
We define the following events:

\begin{equation}\label{def:events}
\left\{
\begin{array}{ccl}
F & = & \text{maximum likelihood fails} \\
F_A & = & \exists_{i\in A} : \text{ i is connected to more nodes in B than in A} \\
\Delta & = & \text{no node in H is connected to at least }\delta(n) \text{ other nodes in H} \\
F_H^{(j)} & = & \text{node }j\in H \text{ satisfies } E(j,A\setminus H) +\delta(n) \leq E(j,B) \\
F_H & = & \cup_{j\in H} F_H^{(j)},
\end{array}
\right.
\end{equation}
where $E(\cdot,\cdot)$ is the number of edges between two sets. Note that we identify nodes of our graph with integers with a slight abuse of notation when there is no risk of confusion.

We also define
\begin{equation}\label{definitionofrho}
\rho(n) = \p\left( F_H^{(i)} \right)
\end{equation}

\begin{lemma}\label{lemma_FAF}
If $\p\left(F_A\right)\geq \frac23$ then $\p\left( F \right)\geq \frac13$.
\end{lemma}
\begin{proof}
By symmetry, the probability of a failure in $B$ is also at least $\frac23$ so, by union bound, with probability at least $\frac13$ both failures will happen simultaneously which implies that ML fails.
\end{proof}

\begin{lemma}\label{lemma:reverseunionboundH}
If $\p\left( F_H\right) \geq \frac9{10}$ then $\p\left( F \right)\geq \frac13$.
\end{lemma}
\begin{proof}
It is easy to see that $\Delta \cap F_H \Rightarrow F_A$ and
Lemma~\ref{lemma:delta910} states that
\begin{equation}
\p\left( \Delta\right) \geq \frac9{10}.
\end{equation}
Hence,
\[
\p\left(F_A\right) \geq \p\left(F_H\right) + \p\left(\Delta\right) - 1 \geq \frac8{10} > \frac23,
\]
which together with Lemma \ref{lemma_FAF} concludes the proof.
\end{proof}

\begin{lemma}\label{lemma:controlrhoisenough}
Recall the definitions in (\ref{def:events}) and (\ref{definitionofrho}). If $$\rho(n) > n^{-1}\gamma(n)\log(10)$$ then, for sufficiently large $n$, $\p\left( F \right)\geq \frac13$.
\end{lemma}

\begin{proof} We will use Lemma \ref{lemma:reverseunionboundH} and show that if $\rho(n) > n^{-1}\gamma(n)\log(10)$ then $\p\left( F_H\right) \geq \frac9{10}$, for sufficiently large $n$.

$F_H^{(i)}$ are independent and identically distributed random variables so
\[
\p\left(F_H\right) = \p\left(\cup_{i\in H} F_H^{(i)}\right) = 1-\p\left(\cap_{i\in H} \left(F_H^{(i)}\right)^c\ \right) = 1-\left( 1 - \p\left(F_H^{(i)}\right) \right)^{|H|} = 1-\left( 1 - \rho(n) \right)^{\frac{n}{\gamma(n)}}
\]
This means that $\p\left( F_H\right) \geq \frac9{10}$ is equivalent to $\left( 1 - \rho(n) \right)^{\frac{n}{\gamma(n)}} \leq \frac1{10}$. If $\rho(n)$ is not $o(1)$ than the inequality is obviously true, if $\rho(n) = o(1)$ then,
\[
\lim_{n\to\infty} \left( 1 - \rho(n) \right)^{\frac{n}{\gamma(n)}} = \lim_{n\to\infty}\left( 1 - \rho(n) \right)^{\frac1{\rho(n)}\rho(n)\frac{n}{\gamma(n)}} = \lim_{n\to\infty}\exp\left( - \rho(n)\frac{n}{\gamma(n)} \right) \leq \frac1{10},
\]
where the last inequality used the hypothesis $\rho(n) > n^{-1}\gamma(n)\log(10)$.
\end{proof}

\begin{definition}\label{def:definitionofT}
Let $N$ be a natural number, $p,q\in [0,1]$, and $\epsilon\geq 0$, we define
\begin{align}
T(N,p, q , \eps) &= \p\left( \sum_{i=1}^N (Z_i - W_i) \geq \eps \right),
\end{align}
where $W_1,\dots, W_N$ are i.i.d.\ Bernoulli$(p)$ and $Z_1,\dots, Z_N$ are i.i.d.\ Bernoulli$(q)$, independent of $W_1,\dots, W_N$.
\end{definition}

%

\begin{lemma}\label{lemma-conv}
Let $\alpha>\beta>0$, then
\begin{align}
 -\log T\left(\frac{n}2,\pp, \qq , \frac{\log(n)}{\log\log(n)}\right)  &\leq   \left( \frac{\alpha + \beta}2 - \sqrt{\alpha \beta} \right) \log(n) +o\left(\log (n)\right). 
\end{align}
\end{lemma}

\begin{proof}[Proof of Theorem~\ref{lemma-conv}]
From the definitions in (\ref{def:events}) and (\ref{definitionofrho}) we have 
\begin{align}
\rho(n)=\p\left( \sum_{i=1}^{\frac{n}2} Z_i -  \sum_{i=1}^{\frac{n}2- \frac{n}{\gamma(n)}} W_i \geq \frac{\log(n)}{\log\log(n)} \right) 
\end{align}
where $W_1,\dots, W_N$ are i.i.d.\ Bernoulli$\left(\pp \right)$ and $Z_1,\dots, Z_N$ are i.i.d.\ Bernoulli$\left(\qq \right)$, all independent.
Since
\begin{align}
\p\left( \sum_{i=1}^{\frac{n}2} Z_i -  \sum_{i=1}^{\frac{n}2- \frac{n}{\gamma(n)}} W_i \geq \frac{\log(n)}{\log\log(n)} \right) \geq \p\left( \sum_{i=1}^{\frac{n}2} Z_i -  \sum_{i=1}^{\frac{n}2} W_i \geq \frac{\log(n)}{\log\log(n)} \right) ,
\end{align}
we get
\begin{align}
- \log \rho(n) \leq - \log T \left(n/2, \pp, \qq, \frac{\log(n)}{\log\log(n)}\right),
\end{align}
and Lemma~\ref{lemma-conv} implies 
\begin{align}
- \log \rho(n) \leq  \left( \frac{\alpha + \beta}2 - \sqrt{\alpha \beta} \right) \log(n) + o(\log(n)).
\end{align}
Hence $\rho(n) > n^{-1}\gamma(n)\log(10)$, and the conclusion follows from Lemma \ref{lemma:controlrhoisenough}.
\end{proof}

\section{Information theoretic upper bound}

We present now the main result of this Section.

\begin{theorem}\label{theorem:main_upperbound_2}
If $\frac{\alpha+\beta}{2} -\sqrt{\alpha \beta} > 1$, i.e., if $\alpha + \beta > 2$ and $(\alpha-\beta)^2 > 4(\alpha+\beta) - 4$, then the maximum likelihood estimator exactly recovers the communities (up to a global flip), with high probability.
\end{theorem}

The case $\beta = 0$ follows directly from the connectivity threshold phenomenon on Erd\H{o}s-R\'enyi graphs so we will restrict our attention to $\alpha> \beta > 0$. 

We will prove this theorem through a series of lemmas. The techniques are similar to the estimates for decoding a codeword on a memoryless channel with a specific structured codes.  In what follows we refer to the true community partition as the ground truth.

\begin{lemma}
If the maximum likelihood estimator does not coincide with the ground truth, then there exists $1\leq k\leq \frac{n}4$ and a set $A_w\subset A$ and $B_w\subset B$ with $|A_w| = |B_w| = k$ such that
\[
E(A_w,B\setminus B_w)+E(B_w,A\setminus A_w) \geq E(A_w,A\setminus A_w)+E(B_w,B\setminus B_w).
\]
\end{lemma}

\begin{proof}
Recall that the maximum likelihood estimator finds two equally sized communities (of size $\frac{n}2$ each) that have the minimum number of edges between them, thus for it to fail there must exist another balanced partition of the graph with a smaller cut, let us call it $Z_A$ and $Z_B$. Without loss of generality $Z_A\cap A \geq \frac{n}4$ and $Z_B \cap B \geq \frac{n}4$. Picking $A_w = Z_B\cap A$ and $B_w = Z_A \cap B$ gives the result.
\end{proof}

Let $F$ be the event of the maximum likelihood estimator not coinciding with the ground truth. Given $A_w$ and $B_w$ both of size $k$, define $P_n^{(k)}$ as
\begin{equation}\label{eq:def:Pnk_1}
P_n^{(k)} := \p\left( E(A_w,B\setminus B_w)+E(B_w,A\setminus A_w) \geq E(A_w,A\setminus A_w)+E(B_w,B\setminus B_w) \right).
\end{equation}

We have, by a simple union bound argument,
\begin{equation}\label{eq:pF_unionbound_ITub}
\p(F) \leq \sum_{k=1}^{n/4} {n/2\choose k}^2 P_n^{(k)}.
\end{equation}

Let $W_i$ be a sequence of i.i.d. Bernoulli$\left(\frac{\alpha \log n}n \right)$ random variables and $Z_i$ an independent sequence of i.i.d. Bernoulli$\left(\frac{\beta \log n}n \right)$ random variables, note that (cf.\ Definition~\ref{def:definitionofT}),
\[
P_n^{(k)} = \p\left( \sum_{i=1}^{2k\left(\frac{n}2-k\right)}Z_i \geq \sum_{i=1}^{2k\left(\frac{n}2-k\right)}W_i  \right) = T\left( 2k\left(\frac{n}2-k\right),\frac{\alpha\log n}{n},\frac{\beta\log n}{n},0\right).
\]

Lemma \ref{lemma:pnk_fromCLT_ITub} in the Appendix shows that:
\begin{equation}\label{eq:comesfromBernstein}
P_n^{(k)} \leq  \exp\left( - \frac{\log(n)}{n}\cdot 4k \left( \frac{n}{2}-k \right) \left( \frac{\alpha+\beta}{2} -\sqrt{\alpha \beta}\right) \right).
\end{equation}

We thus have, combining (\ref{eq:pF_unionbound_ITub}) and (\ref{eq:comesfromBernstein}), and using ${n\choose k} \leq (ne/k)^k$,
\begin{eqnarray}
\p(F) & \leq & \sum_{k=1}^{n/4} {n/2\choose k}^2 \exp\left(  - \frac{\log(n)}{n}\cdot 4k \left( \frac{n}{2}-k \right) \left( \frac{\alpha+\beta}{2} -\sqrt{\alpha \beta}\right)  \right)\nonumber\\
& \leq & \sum_{k=1}^{n/4}  \exp\left( 2k\left(\log\left(\frac{n}{2k} \right)+1\right)  - \frac{\log(n)}{n}\cdot 4k \left( \frac{n}{2}-k \right) \left( \frac{\alpha+\beta}{2} -\sqrt{\alpha \beta}\right)  \right)\nonumber\\
& = & \sum_{k=1}^{n/4}  \exp\left[k\left( 2 \log n - 2\log 2k + 2 - \left(\frac{1}2-\frac{k}n\right) \cdot 4\left( \frac{\alpha+\beta}{2} -\sqrt{\alpha \beta}\right) \log(n) \right)\right]. \label{eq:outsidelemmaboundpf5_IT_ub}
\end{eqnarray}

\begin{proof}[Proof of Theorem~\ref{theorem:main_upperbound_2}]
Recall that $F$ is the event of the maximum likelihood estimator not coinciding with the ground truth. We next show that for $\epsilon>0$, if,
\[
\frac{\alpha+\beta}{2} -\sqrt{\alpha \beta} \geq 1+\epsilon
\]
then there exists a constant $c>0$ such that
\begin{equation}\label{eq:with14epsilon_14}
\p(F) \leq cn^{-\frac14\epsilon}.
\end{equation}

Combining (\ref{eq:outsidelemmaboundpf5_IT_ub}) and (\ref{eq:with14epsilon_14}), we have
\begin{eqnarray}
\p(F) & \leq &  \sum_{k=1}^{n/4}  \exp\left[k\left( 2 \log n - 2\log 2k - \left(\frac{1}2-\frac{k}n\right) \left( 4+\epsilon \right) \log n +2\right)\right]\nonumber\\
& = &  \sum_{k=1}^{n/4}  \exp\left[k\left(  - 2\log 2k +  4\frac{k}n\log n - \left(\frac{1}2-\frac{k}n\right) \epsilon \log n+2 \right)\right]\nonumber\\
& \leq &  \sum_{k=1}^{n/4}  \exp\left[k\left(  - 2\log 2k +  4\frac{k}n\log n - \frac{1}4\epsilon \log n+2 \right)\right]\\
& = &  \sum_{k=1}^{n/4} n^{-\frac{k}4\epsilon} \exp\left[-2k\left( \log 2k -  \frac{2k}n\log n +1 \right)\right].\nonumber
\end{eqnarray}
Note that, for sufficiently large $n$, $1\leq k\leq \frac{n}4$ we have
\[
\log 2k -  \frac{2k}n\log n \geq \frac13 \log(2k),
\]
and $ n^{-\frac{k}4\epsilon} \leq n^{-\frac{1}4\epsilon}$.
Hence, for sufficiently large $n$,
\[
\p(F) \leq n^{-\frac{1}4\epsilon}\sum_{k=1}^{n/4} \exp\left[-\frac23k\left(\log 2k-3\right)\right],
\]
which, together with the observation that $\sum_{k=1}^{n/4} \exp\left[-\frac23k\left(\log 2k-3\right)\right] = O(1)$, concludes the proof of the theorem.
\end{proof}

\section{Efficient algorithms}

\subsection{A semidefinite programming based relaxation}

We propose and analyze an algorithm, based in semidefinite programming (SDP), to efficiently reconstruct the two communities. Let $\G=(V,E(\G))$ be the observed graph, where edges are independently present, with probability $\frac{\alpha \log(n)}{n}$ if they connect two nodes in the same community and with probability $\frac{\beta \log(n)}{n}$ if they connect two nodes in different communities, with $\alpha>\beta$. Recall that there are n nodes in this graph and that with a slight abuse of notation, we will identify nodes in the graph by an integer in $[n]$. Our goal is to recover the two communities in $\G$.\\

The proposed reconstruction algorithm will try to find two communities such that the number of within-community edges minus the across-community edges is largest. We will identify a choice of communities by a vector $x\in\mathbb{R}^n$ with $\pm1$ entries such that the $i^{th}$ component will correspond to $+1$ if node $i$ is in one community and $-1$ if it is in the other. We will also define $B$ as the $n\times n$ matrix with zero diagonal whose non diagonal entries are given by
\[
B_{ij} = \left\{ \begin{array}{l} 1 \text{ if } (i,j)\in E(\G) \\  -1 \text{ if } (i,j)\notin E(\G), \end{array}  \right.
\] 
The proposed algorithm will attempt to maximize the following
\begin{align}
\max\ &x^TBx\\
\text{s.t.}\ &x_i=\pm1.
\end{align}

Our approach will be to consider a simple SDP relaxation to this combinatorial problem. The SDP relaxation considered here dates back to the seminal work of Goemans and Williamson~\cite{MXGoemans_DPWilliamson_1995} on the \texttt{Max-Cut} problem. The techniques behind our analysis are similar to the ones used by the first two authors on a recent publication~\cite{Abbe_Z2SynchJournal,Abbe_Z2SynchER}:

\begin{align}
\max\ & \tr (BX)\nonumber\\
\text{s.t.}\ &X_{ii}=1\label{SDP:mainSDP}\\
& X\succeq 0.\nonumber
\end{align}

\begin{theorem}\label{theorem:main_SDP_provenintheappendix}
If $(\alpha-\beta)^2 > 8(\alpha+\beta) + \frac83(\alpha-\beta)$, the following holds with high probability:
(\ref{SDP:mainSDP}) has a unique solution which is given by the outer-product of $g\in\left\{\pm1\right\}^n$ whose entries corresponding to community $A$ are $1$ and community $B$ are $-1$. 
Hence, if $(\alpha-\beta)^2 > 8(\alpha+\beta) + \frac83(\alpha-\beta)$, full recovery of the communities is possible in polynomial time.
\end{theorem}

We will prove this result through a series of lemmas. Recall that $\G$ is the observed graph and that the vector $g$ corresponds to the correct choice of communities. As stated above, the optimization problem (\ref{SDP:mainSDP}) is an SDP (Semidefinite Program) and any SDP can be solved in polynomial time using methods such as the Interior Point Method. Hence if we can prove that the solution of (\ref{SDP:mainSDP}) is $g$, then we will have proved that the algorithm can recover the correct choice of communities in polynomial time.\\

Recall that the degree matrix $D$ of a graph $G$ is a diagonal matrix where each diagonal coefficient $D_{ii}$ corresponds to the number of neighbours of vertex $i$ and that $\lambda_2(M)$ is the second smallest eigenvalue of a symmetric matrix $M$. 

\begin{definition}
Let $\Gp$ (resp. $\Gm$) be a subgraph of $\G$ that includes the edges that link two nodes in the same community (resp. in different communities) and $A$ the adjacency matrix of $\G$. We denote by $\Dp$ (resp. $\Dm$) the degree matrix of $\Gp$ (resp. $\Gm$) and define the Stochastic Block Model Laplacian to be 
\begin{align*}
L_{SBM}=\Dp-\Dm-A
\end{align*}
\end{definition}

\begin{lemma} \label{lemma:lemma1SDP}
If
\begin{equation}
2L_{SBM} +I_n-\1\1^T \succeq 0 \text{ and } \lambda_2\left(2L_{SBM} +I_n-\1\1^T\right)>0 \label{SDP:condition}
\end{equation} 
then $gg^T$ is the unique solution to the SDP (\ref{SDP:mainSDP}).
\end{lemma}
\begin{proof}
We can suppose that $g=(1,...,1,-1,...,-1)^T$ WLOG.
First of all, we obtain a sufficient condition for $gg^T$ to be a solution to SDP (\ref{SDP:mainSDP}) by using the KKT conditions. This will give us the first part of condition (\ref{SDP:condition}).
The primal problem of SDP (\ref{SDP:mainSDP}) is
\begin{align}
\max\ & \tr (BX)\nonumber\\
\text{s.t.}\ &X_{ii}=1 \nonumber \\
& X\succeq 0.\nonumber
\end{align}

The dual problem of SDP (\ref{SDP:mainSDP}) is
\begin{align}
\min\ & \tr (Y)\nonumber\\
\text{s.t.}\ &Y \succeq B \label{SDP:dual1}\\
& Y \text{ diagonal}. \nonumber
\end{align}

$gg^T$ is guaranteed to be an optimal solution to SDP (\ref{SDP:mainSDP}) under the following conditions:
\begin{itemize}
\item $gg^T$ is a feasible solution for the primal problem
\item There exists a matrix $Y$ feasible for the dual problem such that $\tr(Bgg^T)= \tr(Y)$.
\end{itemize}
The first point being trivially verified, it remains to find such a $Y$ (known as a dual certificate). Generally, one can also use complementary slackness to help find such a certificate but, in this case, it is equivalent to strong duality.

Define a correct (resp. incorrect) edge to be an edge between two nodes in the same (resp. different) community and a correct (resp. incorrect) non-edge to be the absence of an edge between two nodes in different (resp. same) communities. Notice that $(Bgg^T)_{ii}$ counts positively the correct edges and non-edges incident from node $i$ and negatively incorrect edges and incorrect non edges incident from node $i$. In other words
\begin{align}
(Bgg^T)_{ii} &= \text{correct edges + correct non edges - incorrect edges - incorrect non edges}\\
&= (\Dp)_{ii}+\left(\frac{n}{2}-(\Dm)_{ii}\right) - \left(\frac{n}{2}-1-(\Dp)_{ii}\right) - (\Dm)_{ii}\\
&=2\left((\Dp)_{ii}-(\Dm)_{ii}\right)+1
\end{align}
Hence: $\tr\left(Bgg^T\right)=\tr\left(2\left(\Dp-\Dm\right)+I_n\right)$ so $Y=2\left(\Dp-\Dm\right)+I_n$ verifies $\tr(Bgg^T)=\tr(Y)$ and, thus defined, is diagonal. As long as $2\left(\Dp-\Dm\right)+I_n \succeq B$, or in other words, $2L_{SBM} +I_n-\1\1^T \succeq 0$, we can then conclude that $gg^T$ is an optimal solution for SDP (\ref{SDP:mainSDP}). \\

The second part of condition (\ref{SDP:condition}) ensures that $gg^T$ is the unique solution to SDP (\ref{SDP:mainSDP}). Suppose that $X^*$ is another optimal solution to SDP (\ref{SDP:mainSDP}), then $\tr\left(X^*\left(2\left(\Dp-\Dm\right)+I_n-B\right)\right)=Tr \left( X^*\left(2L_{SBM} +I_n-\1\1^T\right)\right)=0$ from complementary slackness and $X^* \succeq 0$.  By assumption, the second smallest eigenvalue of  $2L_{SBM} +I_n-\1\1^T$ is non-zero. This entails that $g$ spans all of its null space. Combining this with complementary slackness, the fact that $X^* \succeq 0$ and $2L_{SBM} +I_n-\1\1^T \succeq 0$, we obtain that $X^*$ needs to be a multiple of $gg^T$. Since $X^*_{ii}=1$ we must have $X^*=gg^T$.
\end{proof}

\begin{proof}[Proof of Theorem~\ref{theorem:main_SDP_provenintheappendix}]

Given Lemma (\ref{lemma:lemma1SDP}), the next natural step would be to control the eigenvalues of $ 2L_{SBM} +I_n-\1\1^T$  when $n \rightarrow \infty$. We want to use Benstein's inequality to do this; to make its application easier, we rewrite $2L_{SBM} +I_n-\1\1^T$ as a linear combination of elementary deterministic matrices with random coefficients.
Define
\begin{align}
\alpha^+_{ij}=\begin{cases} 1 & \text{wp }\A \\ -1 & \text{wp } 1-\A \end{cases}\\
\alpha^-_{ij}=\begin{cases} 1 & \text{wp }\B \\ -1 & \text{wp } 1-\B \end{cases}
\end{align}
where the $(\alpha_{ij}^+)_{i,j}$, $(\alpha_{ij}^-)_{i,j}$  are independent and independent of each other.
Define
\begin{align}
\Delta^+_{ij}&= (e_i-e_j)(e_i-e_j)^T\\
\Delta^-_{ij}&=-(e_i+e_j)(e_i+e_j)^T
\end{align}
where $e_i$ (resp. $e_j$) is the vector of all zeros except the $i^{th}$ (resp. $j^{th}$) coefficient which is 1. 
Using these definitions, we can then write $2L_{SBM} +I_n-\1\1^T$ as the difference of two matrices $C$ and $\Gamma$ where $\Gamma$ is a zero-expectation matrix and $C$, a deterministic matrix that corresponds to the expectation, ie
\begin{align}
2L_{SBM} +I_n-\1\1^T &= \sum_{i<j, j \in S(i)} \alpha^+_{ij} \Delta^+_{ij}+\sum_{i<j, j \notin S(i)} \alpha^-_{ij} \Delta^-_{ij}\\
&= C - \Gamma
\end{align}
where
\begin{align}
C &=  \sum_{i<j, j \in S(i)}\left(2\A-1\right)\Delta^+_{ij}+\sum_{i<j, j \notin S(i)} \left(2\B-1\right) \Delta^-_{ij}\\
\Gamma &=  \sum_{i<j, j \in S(i)}\left(\left(2\A-1\right) - \alpha^{+}_{ij}\right) \cdot \Delta^+_{ij} + \sum_{i<j, j \notin S(i)} \left(\left(2\B-1\right) - \alpha^{-}_{ij}\right) \cdot \Delta^-_{ij}
\end{align}
Notice that $\E[\ap]=2\A-1$ and $\E[\am]=2\B-1$, hence $\E[ \Gamma]=0$.\\

Condition (\ref{SDP:condition}) is then equivalent to
\begin{align} \label{SDP:condition2}
C-\Gamma \succeq 0 \text{ and } \lambda_{\min}\left( C^{\perp g}-\Gamma^{\perp g}\right)>0 \text{ w.h.p.\ }
\end{align}
where $\Gamma^S$ (resp. $C^S$) represents the projection of $\Gamma$ (resp. $C$) onto the space $S$. Typically, if we want to project $\Gamma$ onto the space spanned by the vector $v$, then the projection matrix would be $\Pi=\frac{vv^T}{\|v\|_{2}^2}$ and $\Gamma^{v}=\Pi^T \Gamma \Pi$. $C$ being determinstic, condition (\ref{SDP:condition2}) amounts to controlling the spectral norm of $\Gamma$. This is what is exploited in Lemma \ref{lemma:lemma2SDP} in the appendix where it is shown that condition (\ref{SDP:condition2}) is verified if $\p\left(\lambda_{\max}(\Gamma^\1) \geq n-2\beta\log(n)\right)<n^{-\epsilon}$ and $\p\left(\lambda_{\max}(\Gamma^{\perp \1}) \geq (\alpha-\beta) \log(n)\right)< n^{-\epsilon}$ for some $\epsilon>0$.\\

Using Bernstein to conclude, Lemma \ref{lemma:proj1} in the appendix shows that $\p\left(\lambda_{\max}(\Gamma^\1) \geq n-2\beta\log(n)\right)<n^{-\epsilon}$  for some $\epsilon >0$ when n is big enough and Lemma \ref{lemma:projorth1} in the appendix shows that  $\p\left(\lambda_{\max}(\Gamma^{\perp \1}) \geq (\alpha-\beta) \log(n)\right)< n^{-\epsilon}$ for some $\epsilon>0$ if $(\alpha-\beta)^2>8(\alpha+\beta)+\frac{8}{3}(\alpha-\beta)$. This concludes the proof of the theorem.

\end{proof}

\subsection{Efficient full recovery from efficient partial recovery}\label{black}

In this section we show how to leverage state of the art algorithms for partial recovery in the sparse case in order to construct an efficient algorithm that achieves exact recovery down to the optimal information theoretical threshold.\\

The algorithm proceeds by splitting the information obtained in the graph into a part that is used by the partial recovery algorithm and a part that is used for the local steps. In order to make the two steps (almost) independent, we propose the following procedure: First take a random partition of the edges of complete graph on the $n$ nodes into 2 graphs $H_1$ and $H_2$ (done independently of the observed graph $\G$). $H_1$ is an Erdos-Renyi graph on n nodes with edge probability $C/\log(n)$, $H_2$ is the complement of $H_1$. We then define $G_1$ and $G_2$ subgraphs of $\G$ as $G_1=H_1 \cap \G$ and $G_2=H_2 \cap \G$. In the second step, we apply Massoulie's~\cite{Massoulie_SBM} algorithm for partial recovery to $G_1$. As $G_1$ is an SBM graph with parameters $(C\alpha,C\beta)$, this algorithm is guaranteed~\cite{Massoulie_SBM} to output, with high probability, a partition of the n nodes into two communities $A'$ and $B'$, such that the partition is correct for at least $(1-\delta(C))n$ nodes, where $\delta(C)\to 0$ as $C\to\infty$. In other words, $A'$ and $B'$ coincide with $A$ and $B$ (the correct communities) on at least $(1-\delta(C))n$ nodes. 
Lastly, we flip some of the nodes' memberships depending on the edges they have in $G_2$. Using the communities $A'$ and $B'$ obtained in the previous step, we flip the membership of a given node if it has more edges in $G_2$  going to the opposite community than it has to its own. If the the number of flips in each cluster is not the same, keep the clusters unchanged. 

\begin{figure}
        \centering
        \begin{subfigure}[b]{0.33\textwidth}
                \includegraphics[width=0.85\textwidth]{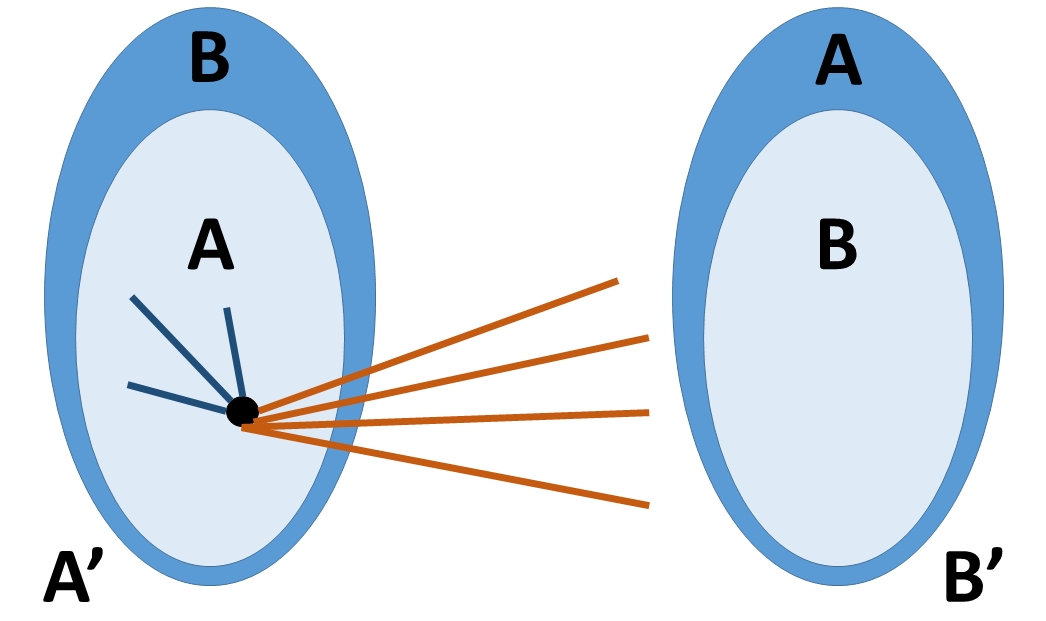}
                \caption{Correct node that will be flipped}
                \label{fig:Correctl}
        \end{subfigure}%
\qquad
        ~ 
        \begin{subfigure}[b]{0.33\textwidth}
                \includegraphics[width=0.85\textwidth]{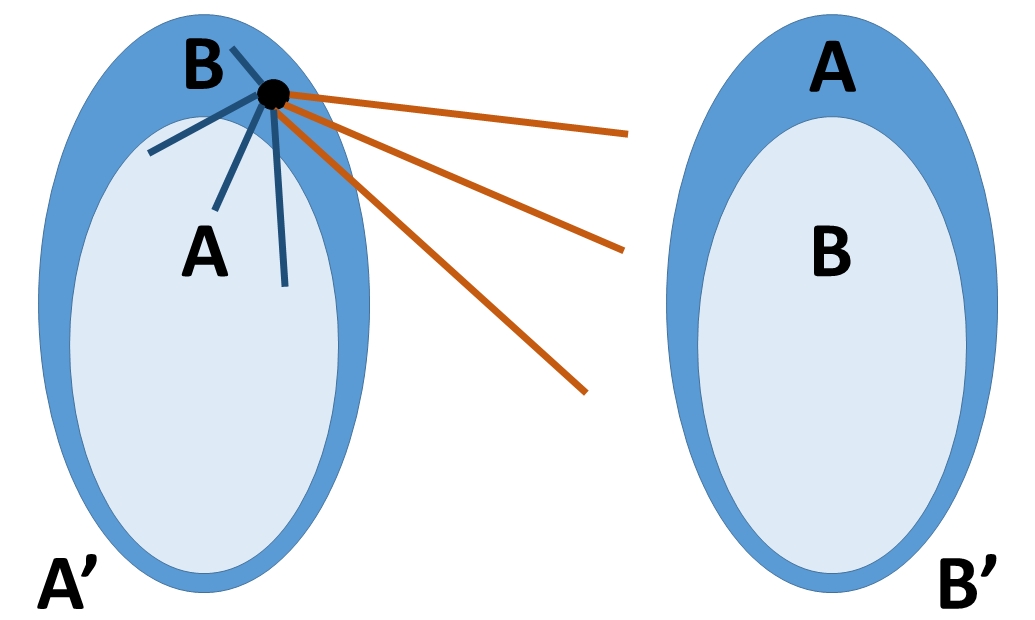}
                \caption{Incorrect node that will not be flipped}
                \label{fig:Incorrect}
        \end{subfigure}
        \caption{Two cases where a node in the graph will be mislabeled}\label{fig:Failures}
\end{figure}

\begin{theorem}
If $\frac{\alpha+\beta}{2}-\sqrt{\alpha \beta}>1$, then, there exists large enough $C$ (depending only on $\alpha$ and $\beta$) such that, with high probability, the algorithm described above will successfully recover the communities from the observed graph.
\end{theorem}

\begin{proof}
In the following, we will suppose that the partial recovery algorithm succeeds as described above w.h.p.\  and we want to show that when $\frac{\alpha+\beta}{2}-\sqrt{\alpha \beta} > 1$ and $\delta$ small enough, the probability that there exists a node that doesn't belong to the correct community, after the local improvements, goes to 0 when $n \rightarrow \infty$. Our goal is to union bound over all possible nodes. We are thus interested in the probability that a node is mislabeled at the end of the algorithm.\\

Recall the random variables $(W_i)_i$ and $(Z_i)_i$ iid and mutually independent Bernoulli random variables with expectations respectively $\alpha \log(n)/n$ and $\beta \log(n)/n$. $W_i$ represents if there is an edge between two nodes in the same community and $Z_i$ if there is an edge between two nodes in different communities. Define $(W'_i)_i$ and $(Z'_i)_i$ iid copies of $(W_i)_i$ and $(Z_i)_i$. For simplicity, we start by assuming that $H_2$ is the complete graph. In this case we have at most $\delta(C)n$ incorrectly labelled nodes (ie $\delta(C) \frac{n}{2}$ nodes that are in A but belong to B' and $\delta(C) \frac{n}{2}$ nodes that are in B but belong to A'). A node in the graph is mislabeled only if it has at least as many connections to the wrong cluster as connections to the right one. This is illustrated in Figure \ref{fig:Failures}. We can express the event with the random variables $(Z_i)_i$, $(W_i)_i$, their copies, and $\delta(C)$.
\begin{align}
P_e = \p(\text{node e is mislabeled}) = \p \left( \sum_{i=1}^{(1-\delta(C))\frac{n}{2}} Z_i +\sum_{i=1}^{\delta(C)\frac{n}{2}}W_i \geq  \sum_{i=1}^{(1-\delta(C))\frac{n}{2}} W'_i +\sum_{i=1}^{\delta(C)\frac{n}{2}}Z'_i \right) \label{probfail}
\end{align}
Recall that we assumed that $H_2$ was a complete graph. In reality, using Lemma \ref{prop:degH1}, it can be shown that the degree of any node in $H_2$ is at least $n \left( 1-2 \frac{C}{\log(n)} \right)$ w.h.p.\ . Taking this into consideration, we will loosely upperbound (\ref{probfail}) by removing $2\frac{C}{\log(n)}n$ on both the rhs terms. Notice that the removal of edges is independent of the outcome of the random variables and
\begin{align}
P_e \leq \p \left( \sum_{i=1}^{(1-\delta(C))\frac{n}{2}} Z_i +\sum_{i=1}^{\delta(C)\frac{n}{2}}W_i \geq  \sum_{i=1}^{(1-\delta(C))\frac{n}{2} - 2 \frac{C}{\log(n)}n} W'_i +\sum_{i=1}^{\delta(C)\frac{n}{2} - 2 \frac{C}{\log(n)}n}Z'_i \right) \label{looseprobfail}
\end{align}
Lemma \ref{failupperbound} shows that (\ref{looseprobfail}) can be upperbounded as follows
\begin{align}
P_e \leq n^{- (g(\alpha,\beta,- \gamma \delta(C) )+o(1))} + n^{-(1+\Omega(1))}.
\end{align}
where 
\begin{align}
&\text{ C is a constant depending only on $\alpha$ and $\beta$ } \\
 &\gamma= \frac{1}{\delta(C) \sqrt{\log(1/\delta(C))}}\\
 &g(\alpha, \beta, \delta')= \frac{\alpha+\beta}{2} - \sqrt{\delta'^2+\alpha \beta} -\delta' \log(\beta) +\frac{\delta'}{2} \log\left(\alpha \beta \cdot \frac{\sqrt{\delta'^2+\alpha \beta }+\delta'}{\sqrt{\delta'^2+\alpha \beta }-\delta'} \right).
\end{align}
Notice that $g(\alpha,\beta,\delta')$ is a function that converges continuously to $f(\alpha,\beta)$ when $\delta' \rightarrow 0$. In this particular case, this is verified as $-\gamma \delta(C) \rightarrow 0$ when $C \rightarrow \infty$. Using a union bound on all nodes
\begin{align}
\p(\exists \text{ mislabeled node }) \leq \sum_{e \in [n]} P_e \leq n^{1- g(\alpha,\beta, -\gamma \delta(C))-o(1)} +n^{-\Omega(1)} \label{fail:last}
\end{align}
For $\delta(C)$ small enough (ie C large enough) and $1-f(\alpha,\beta)<0$, (\ref{fail:last}) goes to 0 when $n \rightarrow \infty$.
\end{proof}

\begin{figure}[h!]
\begin{center}
 \includegraphics[width=0.45\textwidth]{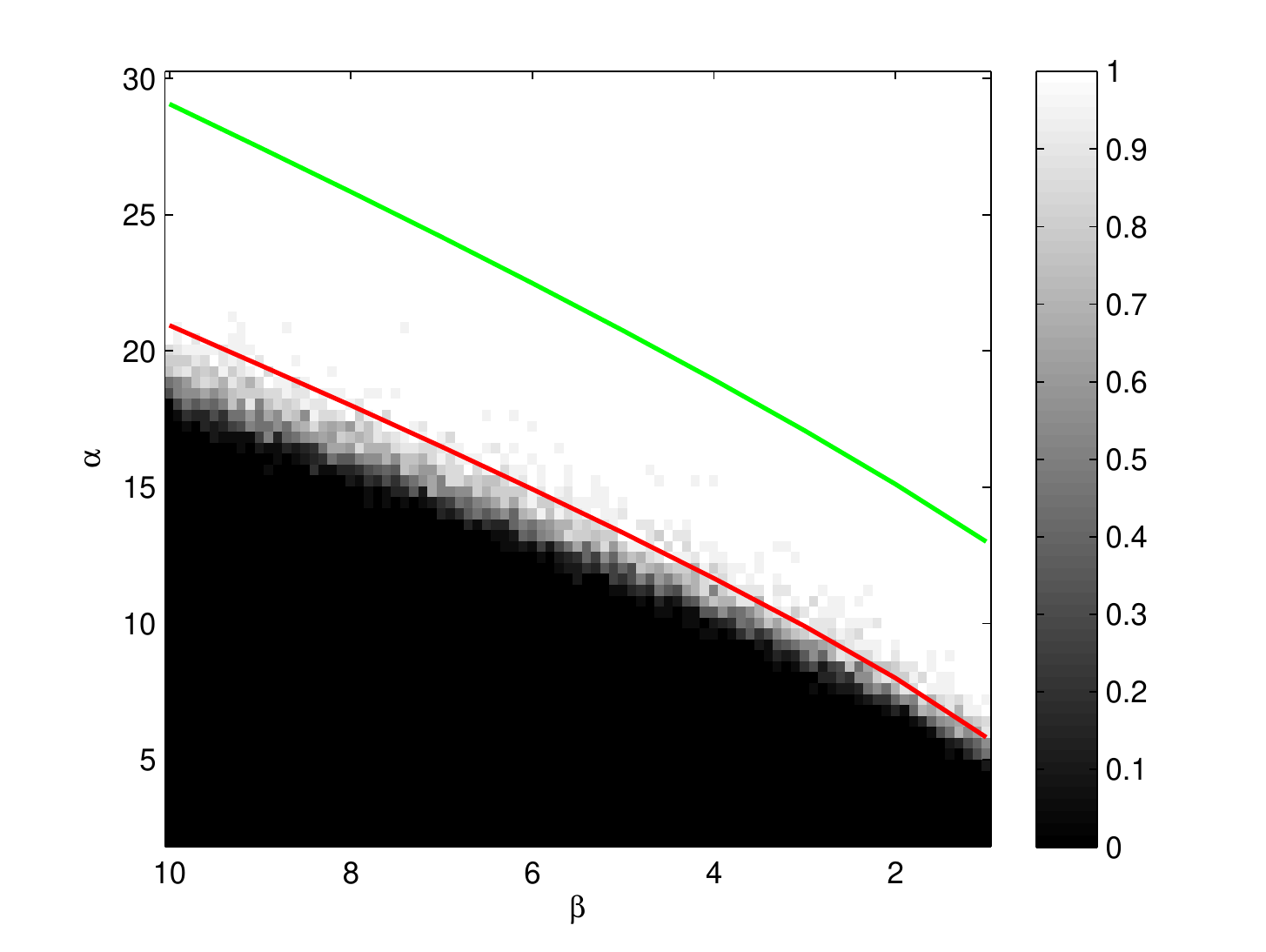}  
\caption{
\label{Figure_Dualcertificate}
{\footnotesize{
This plot shows that the empirical probability of success of the SDP based algorithm essentially matches the optimal threshold of Theorem (\ref{theorem:mainlowerbound}) in red, which is provably achieved with the efficient algorithm of Section \ref{black}. 
We fix $n=300$ and the number of trials to be 20. Then, at each trial and for fixed $\alpha$ and $\beta$, we check how many times each method succeeds. Dividing by the number of trials, we obtain the empirical probability of success by generating the random matrix $C-\Gamma$ corresponding to the correct choice of communities $g=(1,..,1,-1,..,-1)$ and check if condition (\ref{SDP:condition2}) holds (while this implies that the SDP achieves exact recovery it is not necessary). 
In green, we plot the curve corresponding to the threshold given in Theorem (\ref{theorem:main_SDP_provenintheappendix}) ie $(\alpha-\beta)^2 -8(\alpha+\beta)-\frac83(\alpha-\beta)=0$. 
In red, we plot the curve corresponding to the threshold given in Theorems (\ref{theorem:mainlowerbound}) and (\ref{theorem:main_upperbound_2}) ie $(\alpha-\beta)^2 -4(\alpha+\beta)-4=0$ as $\alpha+\beta>2$ in our graph. 
}}\normalsize}
\end{center}
\end{figure}

\section{Conclusion and open problems}
Note that at high SNR (large $\alpha-\beta$), the SDP based algorithm succeeds in the regime of the optimal threshold obtained with ML, up to a factor $2$. When running numerical simulations however, it would seem that the SDP based method achieves exact recovery all the way down to the optimal threshold. As a consequence, the additional factor $2$ is likely a limitation of the analysis, in particular the matrix Bernstein inequality, rather than the algorithm itself.
It remains open to show that this algorithm (or a spectral algorithm) achieves the optimal bound $\frac{\alpha + \beta}2 - \sqrt{\alpha \beta}>1$.  
While we obtain that there is no gap between what can be achieved with an efficient algorithm and the maximum likelihood, as shown in Section \ref{black} using black-box algorithms for partial recovery and local improvement, obtaining direct algorithms would still be interesting. 
It would also be interesting to understand if efficient algorithms achieving the detection threshold can be used to achieve the recovery threshold and vice versa, or whether targeting the two different thresholds leads to different algorithmic developments.

Finally, it is natural to expect that the results obtained in this paper extend to a much more general family of network models, with multiple clusters,  overlapping communities \cite{mixed} and labelled edges \cite{jiaming}. 
   
\subsection*{Acknowledgements}
The authors are grateful to Dustin G. Mixon for thoroughly reading a preliminary version of this manuscript and providing useful comments, as well as to Philippe Rigollet and Van Vu for stimulating discussions on the efficient recovery via partial recovery.

\newpage
\bibliographystyle{plain}
{\small
\bibliography{gen,gen2}}


\newpage
\appendix

\section{Proof of technical lemmas}

\subsection{Tail of the difference between two independent binomials of different parameters}

Recall the following definition:
\begin{definition}\label{def:definitionofT}
Let $m$ be a natural number, $p,q\in [0,1]$, and $\delta\geq 0$, we define
\begin{align}
T(m, p, q , \delta) &= \p\left( \sum_{i=1}^m (Z_i - W_i) \geq \delta \right),
\end{align}
where $W_1,\dots, W_m$ are i.i.d.\ Bernoulli$(p)$ and $Z_1,\dots, Z_m$ are i.i.d.\ Bernoulli$(q)$, independent of $W_1,\dots, W_m$.
\end{definition}
For a better understanding of some of the proofs that follow, it is important to consider the behavior of $T \left( n/2, \alpha \log(n)/n, \beta \log(n)/n, 0 \right)$ when $n \rightarrow \infty$. It can be shown that
\begin{align}
T \left( \frac{n}{2}, \frac{\alpha \log(n)}{n}, \frac{\beta \log(n)}{n}, 0 \right) =\exp \left( - \left( \frac{\alpha + \beta}{2} -\sqrt{\alpha \beta} +o(1) \right) \log(n) \right) \label{behaviorT}
\end{align}
This result is particularly interesting as one can't hope to obtain this bound using standard techniques such as Central Limit Theorem approximations or Chernoff bounds. This comes from the fact that when using these bounds, the error on the exponent is of order $O(\log(n))$ which is relevant here. In the same way, when using an approximation of the binomial coefficients to prove (\ref{behaviorT}), one has to rely on tight estimates. Equation (\ref{behaviorT}) has been extended  to other values of the parameters (typically $T \left(n/2, \alpha \log(n)/n, \beta \log(n)/n,  \epsilon \right)$ where $\epsilon$ is given) but the main idea is contained in equation (\ref{behaviorT}).\\

The idea in the subsequent proofs will be to bound $T(m,p,q,\eps \log(n))$ with its dominant term $T^*(m,p,q,\epsilon)$ that we define below. As a consequence, it is particularly important to bound this dominant term as well, which is what is done in the following lemma. 
\begin{lemma}\label{lemma1}
We recall that $p=\frac{\alpha \log(n)}{n}$ and $q=\frac{\beta \log(n)}{n}$ and we define:
\begin{align*}
V(m,p,q,\tau,\epsilon)=\binom{m}{(\tau+\epsilon) \frac{m}{n} \log(n)} \binom{m}{\tau \frac{m}{n} \log(n)} p^{\frac{m}{n}\tau \log(n)} q^{\frac{m}{n}(\tau +\epsilon) \log(n)}(1-p)^{m-\tau \frac{m}{n}\log(n)} (1-q)^{m-(\tau +\epsilon)\frac{m}{n}\log(n)}
\end{align*}
where $\epsilon=O(1)$. We also define the function
\begin{align} \label{defg}
g(\alpha,\beta,\epsilon)= (\alpha+\beta)-\epsilon \log(\beta) - 2\sqrt{\left( \frac{\epsilon}{2}\right)^2 +\alpha \beta} +\frac{\epsilon}{2} \log \left( \alpha \beta \frac{\sqrt{(\epsilon/2)^2 +\alpha \beta} +\epsilon/2}{\sqrt{(\epsilon/2)^2 +\alpha \beta} -\epsilon/2} \right)
\end{align}
Then we have the following results for $T^*(m,p,q,\epsilon)=\max_{\tau>0} V(m,p,q,\tau,\epsilon)$ :\\

\noindent For $m \in \mathbb{N}$ and $\forall \tau>0$:
\begin{align}
-\log(T^*(m,p,q,\epsilon)) &\geq \frac{m}{n} \log(n) \cdot g(m,n,\epsilon)- o\left(\frac{m}{n}\log(n)\right)  \forall m \in \mathbb{N} \label{lemma:result1} 
\end{align}
For $cn\leq m <c'n^{3/2}$ and $\forall \tau>0$:
\begin{align}
-\log(T^*(m,p,q, \epsilon)) &\leq \frac{m}{n} \log(n) \cdot g(m,n,\epsilon)+o\left(\frac{m}{n}\log(n)\right)  \forall cn\leq m<c'n^{\frac32}  \label{lemma:result2} 
\end{align}
\end{lemma}

\begin{proof}
The proof is mainly computational and the main difficulty comes from upperbounding or lowerbounding the binomial coefficients. We start by writing out $\log(V(m,p,q,\tau,\epsilon))$.

\begin{align*}
\log(V(m,p,q,\tau, \epsilon)=\log \binom{m}{(\tau+\epsilon)\frac{m}{n}\log(n)}+\log \binom{m}{\tau\frac{m}{n}\log(n)} +\frac{m}{n}\tau \log(n) \log(pq)\\
+\frac{m}{n} \epsilon \log(n) \log \left( \frac{q}{1-q}\right)+\left(m -\tau \frac{m}{n}\log(n \right) \log\left((1-p)(1-q)\right)
\end{align*}

In this expression, we replace $p,q$ by their expressions given above and obtain
\begin{align} \label{expressionV1}
\log(V(m,p,q,\tau, \epsilon)&=\log \binom{m}{(\tau+\epsilon)\frac{m}{n}\log(n)}+\log \binom{m}{\tau\frac{m}{n}\log(n)} \nonumber\\
&+\tau \frac{m}{n} \log(n)   \left( \log(\alpha \beta)+2\log \log(n) -2\log(n) \right)\\
&+\epsilon \frac{m}{n} \log(n) \left( \log(\beta)+\log \log(n) -\log(n) +\beta \frac{\log(n)}{n} \right) \nonumber \\
& - \frac{m}{n} \log(n) (\alpha +\beta) + o \left( \frac{m}{n} \log(n) \right) \nonumber
\end{align}

To prove (\ref{lemma:result1}) we upperbound the binomial coefficients using the following result: if $k\leq n$, then $\binom{n}{k} \leq  \left( \frac{ne}{k} \right)^k$ and get
\begin{align}
\log \binom{m}{(\tau +\epsilon) \frac{m}{n}\log(n)} &\leq (\tau +\epsilon) \frac{m}{n} \log(n) \left( \log(n) -\log \log(n) -\log \left( \frac{\tau +\epsilon}{e}\right)\right) \label{bin1}\\
\log \binom{m}{\tau \frac{m}{n} \log(n)} &\leq \tau \frac{m}{n} \log(n) \left( \log(n) -\log \log(n) -\log \left( \frac{\tau}{e}\right)\right) \label{bin2}
\end{align}
We use (\ref{bin1}) and (\ref{bin2}) in (\ref{expressionV1}) and obtain
\begin{align} \label{expressionV2}
-\log(V(m,p,q,\tau, \epsilon)) \geq \frac{m}{n}\log(n) \left( (\alpha +\beta) +(\tau+\epsilon)\log \left( \frac{\tau +\epsilon}{e}\right) +\tau \log \left( \frac{\tau}{e}\right) -\tau \log(\alpha \beta) -\epsilon \log(\beta) \right)\\
 - o \left( \frac{m}{n} \log(n) \right) \nonumber\\ \nonumber
\end{align} 

To prove (\ref{lemma:result2}) we lowerbound the binomial coefficients using the following bound for the binomial coefficient, for $k\leq \sqrt{N}$ (see~\cite{blog:binomial} for a nice presentation) 
\begin{align}
\log \binom{N}{k} \geq k \log(N) -\log(k!) -\log(4)
\end{align}
Merging this inequality with
\begin{align}
\log(k!)\leq (k+1)\log((k+1)/e),
\end{align}
gives
\begin{align}
\log \binom{N}{k} \geq k \log(N) -(k+1)\log(k+1)+k -\log(4)
\end{align}
Given the condition on $m$, we can use the previous inequality and we obtain
\begin{align} \label{binom3}
\log \binom{m}{\tau \frac{m}{n} \log(n)} \geq \tau \frac{m}{n} \log(n) \log(m) - \left(\tau \frac{m}{n} \log(n) +1\right) \log \left( 1+ \tau \frac{m}{n} \log(n) \right) +\tau \frac{m}{n} \log(n) -\log(4)
\end{align}
We expand $\log(1+\tau \frac{1}{n} \log(n))$ as $m\geq c n$ to get
\begin{align} \label{DL1}
\log(1+\tau \frac{m}{n} \log(n)) &=\log(\tau) +\log \left( \frac{m}{n}\right) + \log \log(n) +\frac{1}{\tau \frac{m}{n} \log(n)}+ o\left( \frac{1}{\tau \frac{m}{n} \log(n)} \right)
\end{align}
and replacing (\ref{DL1}) in (\ref{binom3}) we get 
\begin{align} \label{binom4}
\log \binom{m}{\tau \frac{m}{n} \log(n)} \geq \tau \frac{m}{n} \log(n) \left( \log(n)+\log \left( \frac{e}{\tau}\right) -\log \log(n) \right) - o \left( \frac{m}{n} \log(n) \right) 
\end{align}
In the same way
\begin{align} \label{binom5}
\log \binom{m}{(\tau+\epsilon) \frac{m}{n} \log(n)} \geq \tau \frac{m}{n} \log(n) \left( \log(n)+\log \left( \frac{e}{\tau+\epsilon}\right) -\log \log(n) \right) - o \left( \frac{m}{n} \log(n) \right) 
\end{align}
Now using (\ref{binom4}) and (\ref{binom5}) in (\ref{expressionV1}) we get
\begin{align} \label{expressionV3}
-\log(V(m,n,p,q,\tau, \epsilon) \leq \frac{m}{n} \log(n) \left( (\tau +\epsilon) \log \left( \frac{\tau +\epsilon}{e}\right) +\tau \log \left( \frac{\tau}{e}\right) - \tau \log(\alpha \beta) -\epsilon \log(\beta) +(\alpha +\beta) \right)\\
 + o \left( \frac{m}{n}\log(n) \right) \nonumber
\end{align}

\noindent We now consider
\begin{align}
h(\alpha, \beta,\tau,\epsilon)=(\tau +\epsilon) \log \left( \frac{\tau +\epsilon}{e}\right) +\tau \log \left( \frac{\tau}{e}\right) - \tau \log(\alpha \beta) -\epsilon \log(\beta) +(\alpha +\beta)
\end{align}
We minimize $h(\alpha,\beta,\tau,\epsilon)$ with respect to $\tau$. We obtain
\begin{align}
\tau^*= -\frac{\epsilon}{2} +\sqrt{\left(\frac{\epsilon}{2} \right)^2+\alpha \beta}
\end{align}
We replace $\tau$ by $\tau^*$ in (\ref{expressionV2}) and (\ref{expressionV3}) and obtain the results given in the lemma.
\end{proof}

\noindent
{\bf Lemma \ref{lemma-conv}.}
{\it Let $\alpha>\beta>0$, then
\begin{align}
 -\log T\left(\frac{n}2,\pp, \qq , \frac{\log(n)}{\log\log(n)}\right)  &\leq   \left( \frac{\alpha + \beta}2 - \sqrt{\alpha \beta} \right) \log(n) +o\left(\log (n)\right). 
\end{align}}

\begin{proof}
Let $\delta=\delta(n)=\lceil \log(n)/\log\log(n) \rceil$. For sake of brevity we take $p = \frac{\alpha \log n}{n}$ and $q = \frac{\alpha \log n}{n}$. 
 
By definition, $T(n/2,p, q , \delta)$ is larger than the probability that $\sum_{i=1}^{n/2} (Z_i - W_i)$ is equal to $\delta$, hence  
\begin{align}
T(n/2,p, q , \delta) &\geq \sum_{k =0}^{n/2-\delta} {n/2 \choose k}{n/2 \choose k+\delta} p^k (1-p)^{n/2-k} q^{k+\delta} (1-q)^{n/2-k-\delta}.
\end{align}
Choosing $k = \tau \log(n)$, for $\tau >0$, $k$ is in the range $[0,n/2-\delta]$ for $n$ sufficiently large
\begin{align}
&   T\left(\frac{n}2,\pp, \qq , \delta \right) \\
& \geq   \max_{\tau >0} \binom{n/2}{\tau \log(n)} \binom{n/2}{\tau \log(n) +\delta}\left(\pp \right)^{\tau \log(n)}\left( \qq \right)^{\tau \log(n) +\delta} \\
&  \quad \cdot \left(1-\pp\right)^{n/2 - \tau \log(n)} \left(1-\qq\right)^{n/2 - \tau \log(n)-\delta} \label{start}  \\
&= T^* \left( \frac{n}{2}, p,q,\epsilon \right)
\end{align}
where ${n/2 \choose  \tau \log(n)}$ is defined as ${n/2 \choose  \lfloor \tau \log(n) \rfloor }$ if $\tau \log(n)$ is not an integer and $\eps=1/\log \log(n)$.\\

We use the result from Lemma 1 with $m=\frac{n}{2}$ and $\epsilon=1/\log \log(n)$ and notice that
\begin{align*}
&\epsilon \log(\beta) = o(1)\\
&2\sqrt{ \left(\frac{\epsilon}{2} \right)^2+\alpha \beta} =2\sqrt{\alpha \beta} + o(1)\\
&\frac{\epsilon}{2} \log \left( \alpha \beta \frac{\sqrt{(\epsilon/2)^2 +\alpha \beta} +\epsilon/2}{\sqrt{(\epsilon/2)^2 +\alpha \beta} -\epsilon/2} \right) = o(1)
\end{align*}
Hence 
\begin{align*}
-\log(T^* \left( \frac{n}{2}, p,q, \epsilon \right) \leq \left(\frac{\alpha +\beta}{2}-\sqrt{\alpha \beta} \right) \log(n) +o(\log(n))
\end{align*}
and we conclude.
\end{proof}

\begin{lemma}\label{lemma:pnk_fromCLT_ITub} 
Let $W_i$ be a sequence of i.i.d. Bernoulli$\left(\frac{\alpha \log n}n \right)$ random variables and $Z_i$ an independent sequence of i.i.d. Bernoulli$\left(\frac{\beta \log n}n \right)$ random variables. Recall (Definition~\ref{def:definitionofT}) that
\[
T\left( 2k\left(\frac{n}2-k\right),\frac{\alpha\log n}{n},\frac{\beta\log n}{n},0\right) = \p\left( \sum_{i=1}^{2k\left(\frac{n}2-k\right)}Z_i \geq \sum_{i=1}^{2k\left(\frac{n}2-k\right)}W_i  \right).
\]
The following bound holds for $n$ sufficiently large:
\begin{align}
T\left( m,\frac{\alpha\log n}{n},\frac{\beta\log n}{n},0\right)  \leq \exp \left(-\frac{2m}{n} \left( \frac{\alpha+\beta}{2} -\sqrt{\alpha \beta} + o(1)\right) \log(n)\right)
\end{align}
where $m=2k \left( \frac{n}{2}-k\right)$.
\end{lemma}
\begin{proof}
In the following, for clarity of notation, we have omitted the floor/ceiling symbols for numbers that are not integers but should be. 
Recall that
\begin{align}
T(m,p,q,0)=\p[Z-W \geq 0]
\end{align}
where $Z$ is a Binomial$(m,q)$, $W$ is a Binomial$(m,p)$ and $p=\frac{\alpha \log(n)}{n}$, $q=\frac{\alpha \log(n)}{n}$.\\

The idea behind the proof is to bound $\log(T(m,p,q,0))$ with the dominant term $\log(T^*(m,p,q,0)$ when $n$ is large and then use Lemma \ref{lemma1}. Notice that $n-2 \leq m \leq \frac{n^2}{4}$, we split the proof into 2 parts based on the regime of $m$.\\

The first case corresponds to $m$ such that $m \geq n\log\log n$. What is important is that $n = o(m)$. We have
\begin{align}
T(m,p,q,0) &= \sum_{k_1=0}^m \left(\sum_{k_2=k_1}^m \p(Z=k_2) \right) \p(W=k_1)
\end{align}
Notice that each term in the double-sum can be upper-bounded by $T^*(m,p,q,0)$ as defined in (\ref{lemma1}). 
Hence
\begin{align}
T(m,p,q,0) \leq m^2 T^*(m,p,q,0)
\end{align}
and using (\ref{lemma:result1}) for $\epsilon=0$
\begin{align}
-\log(T(m,p,q,0)) &\geq -2\log(m) -\log(T^*(m,p,q,0))\\
& \geq -2\log(m) + \frac{2m}{n} \left( \frac{\alpha+\beta}{2} -\sqrt{\alpha \beta}\right) \log(n)
\end{align}
As $\frac{m}{n}\geq \log \log n$ and $m \leq n^2/4$, notice that $\log(m) = o\left( \frac{m}{n} \log(n) \right)$ and
\begin{align}
-\log(T(m,p,q,0))  \geq \frac{2m}{n} \left( \frac{\alpha+\beta}{2} -\sqrt{\alpha \beta}\right) \log(n) -  o\left( \frac{m}{n} \log(n) \right).
\end{align}

The second case corresponds to $m< n\log \log n $. We define $\delta = \frac{m}{n}$ and note that $\delta < \log \log n$. Notice that the same idea as in the proof above does not work when $m$ is $O(n)$. Nevertheless a similar idea gives valid results by restricting ourselves to the first $\log^2(n)$ terms of the sum over $m$, breaking $T$ as 
\begin{align} \label{Tdef}
T(m,p,q,0) &= \p \left(  0 \leq Z-W \leq \log^2(n) \right)+ \p \left( Z-W \geq \log^2(n) \right).
\end{align}
We want to control both terms in the above sum. We start off by upperbounding $\p \left( Z-W \geq \log^2(n) \right)$ using Bernstein. Let us consider a sequence $X_j$ of $2m$ centered random variables, 
the first $m$ given by $X_j = Z_j -\frac{\beta{\log n}}n$ and the last $m$ by $X_{j+m} = -W_j + \frac{\alpha{\log n}}n$.
Then $Z-W =\sum_{j=1}^{2m} X_j -m(\alpha -\beta) \frac{\log(n)}{n}$ and 
\[
\sum_{i=1}^{2m}\mathbb{E}X_i^2 = m\left[(\alpha+\beta)\frac{\log n}{n} + \OOO\left( \frac{(\log n)^2}{n^2} \right)\right],
\]
Also, $$\left|X_i\right| \leq 1 + \OOO\left( \frac{\log n}{n} \right).$$
We can hence apply Bernestein's inequality and get, for any $t\geq 0$,
\[
\p\left(\sum_{i=1}^{2m} X_i > t \right) \leq \exp\left( - \frac{\frac12t^2}{m(\alpha+\beta)\frac{\log n}{n} + m\,\OOO\left( \frac{(\log n)^2}{n^2} \right)+\frac13t\left(1+\OOO \left( \frac{\log n}{n} \right)\right)} \right).
\]
Here we take $t= m(\alpha-\beta) \frac{\log(n)}{n} + \log^2(n)$
\begin{align} \label{Bst1}
\p \left( Z-W \geq \log^2(n) \right) &=\p \left( \sum_{i=1}^{2m}X_i > m(\alpha-\beta) \frac{\log(n)}{n} + \log^2(n)\right) \nonumber \\
&\leq \exp\left( - \frac{\frac12 \left( \delta \log(n) (\alpha-\beta) +\log^2(n) \right)^2}{\delta(\alpha+\beta)\log(n) + \delta\,\OOO\left( \frac{(\log n)^2}{n} \right)+\frac13 \left( \delta \log(n) (\alpha-\beta) +\log^2(n) \right)\left(1+\OOO \left( \frac{\log n}{n} \right)\right)} \right) \nonumber\\
& \leq \exp\left( - \frac{\frac12 \log^2(n) \left( \delta \frac{ \alpha-\beta}{\log(n)}+ 1\right)^2}{\delta\frac{\alpha+\beta}{\log(n)} + \delta\,\OOO\left( \frac{1}{n} \right)+\frac13 \left( \delta \frac{\alpha-\beta}{\log(n)} + 1\right)\left(1+\OOO \left( \frac{\log n}{n} \right)\right)} \right) \nonumber\\
& \leq \exp \left(- \Omega (1) \frac{\log^2(n)}{\delta}\right)\nonumber\\
& \leq \exp \left(- \Omega (1) \frac{\log^2(n)}{\log \log n}\right).
\end{align}
We now want to control $\p \left(  0 \leq Z-W \leq \log^2(n) \right)$, note that 
\begin{align} \label{smallsum}
\p \left(  0 \leq Z-W \leq \log^2(n) \right) &=\sum_{k_1=0}^{\log^2(n)} \sum_{k_2=0}^{m-k_1} \p(Z=k_1+k_2) \p(W=k_2) \nonumber\\
& \leq \log^2(n)  \left( \sum_{k_2=0}^{\log^2(n)} \p(Z=k_2) \p(W=k_2) + \sum_{k_2= \log^2(n)}^{m} \p(Z=k_2) \p(W=k_2) \right) \nonumber \\
& \leq \log^4(n) T^*(m,p,q,0) + \log^2(n) \p \left( Z \geq \log^2(n) \right) \p \left(W \geq \log^2(n) \right). 
\end{align}
Much as before we use Bernstein inequality to upperbound $\p \left( Z \geq \log^2(n) \right) $ and $\p \left( W \geq \log^2(n) \right) $. Recall that $Z=\sum_{i=1}^m Z_i$ where $Z_i \sim$ Ber$\left(\frac{\beta \log(n)}{n} \right)$. Define $X_i=Z_i -\frac{\beta \log(n)}{n}.$ We have $$\E \left( X_i^2\right) =\frac{\beta \log(n)}{n} + \OOO \left(\frac{\beta \log^2(n)}{n^2}\right) \quad \text{and} \quad |X_i| \leq 1+\OOO\left( \frac{\log(n)}{n} \right).$$
Hence in the same way as in (\ref{Bst1})
\begin{align} \label{Bst2}
\p \left( Z \geq \log^2(n) \right) &= \p \left( \sum_{i=0}^{m} X_i \geq \log^2(n) + m \beta\frac{\log(n)}{n}\right)\\ \nonumber 
&\leq \exp \left(- \Omega (1) \frac{\log^2(n)}{\log \log n}\right), \nonumber
\end{align}
similarly 
\begin{align} \label{Bst3}
\p \left( W \geq \log^2(n) \right) \leq  \exp \left(- \Omega (1) \frac{\log^2(n)}{\log \log n}\right).
\end{align}
Plugging  (\ref{Bst2}), (\ref{Bst3}) into (\ref{smallsum}) we get
\begin{align} \label{laststep}
\p \left(  0 \leq Z-W \leq \log^2(n) \right) \leq \log^4(n) T^*(m,p,q,0) + \log^2(n) \exp \left(- \Omega (1) \frac{\log^2(n)}{\log \log n}\right)
\end{align}
And plugging  (\ref{Bst1}) and (\ref{laststep}) into (\ref{Tdef}) we obtain
\begin{align}
T(m,p,q,0) \leq  \log^4(n) T^*(m,p,q,0) + \log^2(n) e^{ \left( - \Omega(1) \frac{\log^2(n)}{\log \log n}\right)} + e^{ \left( - \Omega(1) \frac{\log^2(n)}{\log \log n}\right)}
\end{align}
From (\ref{lemma:result1}) and $e^{ \left( - \Omega(1) \frac{\log^2(n)}{\log \log n}\right)} = o\left(e^{\log(n)}\right)$  we get
\begin{align}
- \log( T(m,p,q,0)) &\geq  -4\log \log(n) + \frac{2m}{n}\left( \frac{\alpha+\beta}{2} -\sqrt{\alpha \beta}\right) \log(n) - o(\log(n)) \\
& \geq \frac{2m}{n}\left( \frac{\alpha+\beta}{2} -\sqrt{\alpha \beta}\right) \log(n) - o \left(\frac{m}{n} \log(n) \right).
\end{align}
\end{proof}

\begin{lemma}\label{failupperbound}
Let $(W_i)_i$ and $(Z_i)_i$ be iid and mutually independent Bernouillis with expectations respectively $\frac{\alpha \log(n)}{n}$ and $\frac{\beta \log(n)}{n}$. Define $(W'_i)_i$ and $(Z'_i)_i$ iid copies of $(W_i)_i$ and $(Z_i)_i$. Then:
\begin{align*}
 \p \left( \sum_{i=1}^{(1-\delta(C))\frac{n}{2}} Z_i +\sum_{i=1}^{\delta(C)\frac{n}{2}}W_i \geq  \sum_{i=1}^{(1-\delta(C))\frac{n}{2} - 2\frac{C}{\log(n)}n} W'_i +\sum_{i=1}^{\delta(C)\frac{n}{2} - 2 \frac{C}{\log(n)}n}Z'_i \right) \leq n^{- (g(\alpha,\beta, \delta)+o(1))} + n^{-(1+\Omega(1))}
\end{align*}
where $g(\alpha,\beta,\delta)$ is defined in (\ref{defg}).
\end{lemma}
\begin{proof}
Trivially we have
\begin{align}
&\p \left( \sum_{i=1}^{(1-\delta(C))\frac{n}{2}} Z_i +\sum_{i=1}^{\delta(C)\frac{n}{2}}W_i \geq  \sum_{i=1}^{(1-\delta(C))\frac{n}{2} - 2\frac{C}{\log(n)}n} W'_i +\sum_{i=1}^{\delta(C)\frac{n}{2} - 2 \frac{C}{\log(n)}n}Z'_i \right) \\
&\leq \p \left( \sum_{i=1}^{\frac{n}{2}} Z_i +\sum_{i=1}^{\delta(C)\frac{n}{2}}W_i \geq  \sum_{i=1}^{(1-\delta(C))\frac{n}{2} - 2 \frac{C}{\log(n)}n} W'_i \right) \\
&\leq  \p \left( \sum_{i=1}^{\frac{n}{2}} Z_i  - \sum_{i=1}^{\frac{n}{2}}W'_i + \sum_{i=1}^{\delta(C) \frac{n}{2}} W_i + \sum_{i=1}^{\delta(C) \frac{n}{2}+2 \frac{C}{\log(n)}n} W'_i \geq 0 \right) \\
&\leq \p \left( \sum_{i=1}^{\frac{n}{2}} Z_i  - \sum_{i=1}^{\frac{n}{2}}W'_i \geq -\gamma \cdot \delta(C) \log(n) \right) + \p \left( \sum_{i=1}^{\delta(C) \frac{n}{2}} W_i + \sum_{i=1}^{\delta(C) \frac{n}{2}+2 \frac{C}{\log(n)}n} W'_i  \geq \gamma \cdot \delta(C) \log(n) \right) \label{sumoftwoterms}
\end{align}
where $\gamma=\frac{1}{\delta \sqrt{\log(1/\delta)}}$.\\
For the second part of (\ref{sumoftwoterms}), we upperbound using multiplicative Chernoff. Mutliplicative Chernoff states
\begin{align}
\p \left( \sum_{i=1}^{\delta(C)n} W_i \geq (1+\epsilon) \delta(C) \alpha \log(n) \right) \leq \left( \frac{1+\epsilon}{e}\right)^{-(1+\epsilon)\delta(C) \alpha \log(n)}
\end{align}
In our case $1+\epsilon=\frac{\gamma}{\alpha}$. To simplify notation we will write $\delta$ instead of $\delta(C)$ in the following.
\begin{align}
\p \left( \sum_{i=1}^{\delta \frac{n}{2}} W_i + \sum_{i=1}^{\delta \frac{n}{2}+2 \frac{C}{\log(n)}n} W'_i  \geq \gamma \cdot \delta \log(n) \right)  &\leq \p \left(\sum_{i=1}^{\delta n+2 \frac{C}{\log(n)}n} W'_i  \geq \gamma \cdot \delta \log(n) \right)\\
&\leq \p \left(\sum_{i=1}^{\delta n} W'_i  \geq \gamma \cdot \delta \log(n) \right)\\
&\leq \left( \frac{1}{\delta \sqrt{\log(1/\delta)} \cdot \alpha e} \right)^{-\frac{\log(n)}{\sqrt{\log(1/\delta)}}} \\
&\leq n^{-\sqrt{\log(1/\delta)}+ \frac{1}{\sqrt{\log(1/\delta)}} \cdot \left(\log\left(\sqrt{\log(1/\delta)}\right)+\log(\alpha)+ 1 \right)}\\
&\leq n^{-(1 +\Omega(1))}
\end{align}
for small enough $\delta$.

For the first part of (\ref{sumoftwoterms}), we adapt Lemma \ref{lemma:pnk_fromCLT_ITub}.
\begin{align}
\p \left( \sum_{i=1}^{\frac{n}{2}}Z_i -\sum_{i=1}^{\frac{n}{2}} W'_i \geq -\gamma \cdot \delta(C) \log(n) \right) &= T \left( \frac{n}{2},p,q, -\gamma \cdot \delta(C) \log(n) \right) \label{inequality4}\\
&= \p \left( -\gamma \cdot \delta(C) \log(n) \leq Z-W \leq \log(n)^2 \right) +\p \left(Z-W \geq \log(n)^2 \right) \label{inequality2}\\
\end{align}
As shown in  Lemma \ref{lemma:pnk_fromCLT_ITub} the second part of inequality (\ref{inequality2}) can be upperbounded in the following way
\begin{align}
\p \left( Z-W \geq \log(n)^2\right) \leq \exp \left( -\Omega(1) \frac{\log(n)^2}{\log(\log(n))}\right) \label{inequality3}
\end{align}
We now upperbound the first part of inequality (\ref{inequality2}) in a similar way to Lemma \ref{lemma:pnk_fromCLT_ITub}.
\begin{align}
 \p \left( -\gamma \cdot \delta(C) \log(n) \leq Z-W \leq \log(n)^2 \right) \leq \left(\log(n)^2 + \gamma \delta(C) \log(n)\right)^2 T^* \left( \frac{n}{2}, p, q, -\gamma \delta(C) \log(n) \right) \\ \nonumber
+ \log(n)^2 \exp \left( -\Omega(1) \frac{\log(n)^2}{\log(\log(n))}\right) \label{inequality3}
\end{align}
We group inequalities (\ref{inequality2}) and (\ref{inequality3}), we then take the log and using (\ref{lemma:result1}) we obtain
\begin{align}
- \log \left( T\left( \frac{n}{2}, p, q, -\gamma \delta(C) \log(n) \right) \right) \geq \log(n) g(\alpha, \beta, - \gamma \delta(C)) -o(\log(n))
\end{align}
We conclude using (\ref{inequality4})
\begin{align*}
\p \left( \sum_{i=1}^{\frac{n}{2}}Z_i -\sum_{i=1}^{\frac{n}{2}} W'_i \geq -\gamma \cdot \delta(C) \log(n) \right) \leq n^{-g(\alpha, \beta, -\gamma \delta(C)) +o(1)}
\end{align*}

\end{proof}

\subsection{Information Theoretic Lower Bound Proofs}

\begin{lemma}\label{lemma:delta910}
Recall the events defined in (\ref{def:events}). $\p\left( \Delta\right) \geq \frac9{10}$.
\end{lemma}

\begin{proof}
Recall that $\Delta$ is the event that in a graph with $n/\log^3(n)$ vertices where each pair of nodes is connected, independently, with probability $\frac{\alpha\log n}{n}$, every node has degree strictly less than $\frac{\log n}{\log \log n}$. 

Let $\Delta_i$ be the probability that the degree of node $i$ is smaller than $\frac{\log n}{\log \log n}$. Let $X_i$ be iid Bernoulli$\left(\frac{\alpha\log n}{n}\right)$ random variables, then
\[
\p\left( \Delta_i^c \right) = \p\left( \sum_{i=1}^{n/\log^3n-1}X_i \geq \frac{\log n}{\log \log n} \right) \leq \p\left( \sum_{i=1}^{n/\log^3n}X_i \geq \frac{\log n}{\log \log n} \right)
\]
If we set $\mu = \mathbb{E}\left[ \sum_{i=1}^{n/\log^3n}X_i \right] = \frac{n}{\log^3n}\frac{\alpha\log n}n = \alpha\frac1{\log^2n}$, The multiplicative Chernoff bound gives, for any $t>1$,
\[
\p\left( \sum_{i=1}^{n/\log^3n}X_i \geq t \mu \right) \leq \left( \frac{e^{t-1}}{t^t} \right)^\mu.
\]
We consider a slightly weaker version (since $\mu >0$)
\[
\p\left( \sum_{i=1}^{n/\log^3n}X_i \geq t \mu \right) \leq \left( \frac{e^{t-1}}{t^t} \right)^\mu \leq \left( \frac{e^{t}}{t^t} \right)^\mu = \left( \frac{t}e \right)^{-t\mu}.
\] 
This means that, by setting $t = \frac{\log^2 n}{\alpha}\frac{\log n}{\log \log n} =\frac{\log^3 n}{\alpha \log \log n}$ so that $t\mu = \frac{\log n}{\log \log n}$, we have
\[
\p\left( \sum_{i=1}^{n/\log^3n}X_i \geq \frac{\log n}{\log \log n} \right)  \leq  \left( \frac1{e}\left( \frac{\log^3 n}{\alpha \log \log n} \right) \right)^{-\frac{\log n}{\log \log n}}
\]
By union bound we have, for any vertex $i$,
\begin{eqnarray*}
1 - \p\left( \Delta \right) & \leq & \frac{n}{\log^3 n} \p\left( \Delta_i^c \right) \\
 & \leq & \frac{n}{\log^3 n} \left( \frac1{e}\left( \frac{\log^3 n}{\alpha \log \log n} \right) \right)^{-\frac{\log n}{\log \log n}}\\
  & = & \exp\left[ \log\left( \frac{n}{\log^3 n} \right) - \frac{\log n}{\log \log n} \log\left( \frac1{e\alpha}\left( \frac{\log^3 n}{ \log \log n} \right) \right)  \right]\\
  & = & \exp\left[ -2\log n - 3\log\log n + \frac{\log n \log(e\alpha)}{\log \log n}  +\frac{\log n\log\log\log n}{\log \log n}      \right]\\
  & = & \exp\left[ -\left(2 - \OOO\left( \frac{\log\log\log n}{\log\log n} \right) \right)\log n\right], 
\end{eqnarray*}
which proves the Lemma.

\end{proof}

\subsection{SDP Algorithm Proofs}

Recall that $\Gamma^S$ (resp. $C^S$) denotes the projection of $\Gamma$ (resp. C) onto the space S.

\begin{lemma} \label{lemma:lemma2SDP}
If $\p\left(\lambda_{\max}(\Gamma^\1) \geq n-2\beta\log(n)\right)<n^{-\epsilon}$ and $\p\left(\lambda_{\max}(\Gamma^{\perp \1}) \geq (\alpha-\beta) \log(n)\right)< n^{-\epsilon}$ for some $\epsilon>0$ then condition (\ref{SDP:condition2})
\begin{align*}
C-\Gamma \succeq 0 \text{ and } \lambda_{\min}\left( C^{\perp g}-\Gamma^{\perp g}\right)>0
\end{align*}
is verified w.h.p.\ .
\end{lemma}
 \begin{proof} C is the following deterministic symmetric matrix

\begin{equation*}
C=\begin{bmatrix} d & \text{ } &  a & \vrule & \text{ } & \text{ }  & \text{ } \\
\text{ } & \ddots & \text{ } & \vrule & \text{ }&  b & \text{ } \\
a & \text{ } & d & \vrule & \text{ } & \text{ } & \text{ } \\ \hline
 \text{ } & \text{ }  & \text{ } & \vrule & d & \text{ } & a \\
 \text{ }& b & \text{ }  & \vrule & \text{ } & \ddots & \text{ } \\
\text{ } & \text{ } & \text{ }& \vrule & a & \text{ } & d \\
\end{bmatrix}
\end{equation*}
where
\begin{align}
a&=-\frac{2\alpha \log(n)}{n}+1\\
b&=-\frac{2\beta \log(n)}{n}+1\\
d&=(\alpha-\beta)\log(n)-\frac{2\alpha \log(n)}{n}+1
\end{align}

Assuming $\alpha> \beta$ the eigenvalues of C take three distinct values
\begin{itemize}
\item $\lambda_1=n-2\beta \log(n)$ associated to the eigenvector $\1$
\item $\lambda_2=0$ associated to the eigenvector corresponding to the ground truth g
\item $\lambda_3=(\alpha-\beta)\log(n)$ associated to all other eigenvectors
\end{itemize}
g is also an eigenvector for $\Gamma$ corresponding to eigenvalue 0. As a consequence, condition (\ref{SDP:condition2}) is satisfied if the following holds on the orthogonal of g
\begin{align}
\p\left(\lambda_{\min}(C^\1)>\lambda_{\max}(\Gamma^\1)\right)\rightarrow 1 \text{ and } \p\left(\lambda_{\min}(C^{\perp \1})>\lambda_{\max}(\Gamma^{\perp \1})\right)\rightarrow 1 \text{ when } n \rightarrow \infty
\end{align}

This is achieved if
\begin{align}\label{SDP:condition3}
\p\left(n-2\beta \log(n) \leq \lambda_{\max}(\Gamma^{\1})\right)<n^{-\epsilon}  \text{ and } \p\left((\alpha-\beta)\log(n) \leq \lambda_{\max}(\Gamma^{\perp \1})\right)<n^{-\epsilon} \text{ for some } \epsilon>0
\end{align}
\end{proof}

\begin{theorem} \label{theorem:MatrixBernstein}
(Matrix Bernstein) Consider a finite sequence $\{X_k\}$ of independent, random, self adjoint matrices with dimension d. Assume that each random matrix satisfies
\begin{align}
\E X_k=0 \text{ and } \lambda_{\max}(X_k) \leq R \text{ almost surely }
\end{align}
Then, for all $t \geq 0$:
\begin{align}
\p \left( \lambda_{\max}\left(\sum_{k} X_k \right) \geq t \right) \leq d \cdot \exp\left( \frac{-t^2/2}{\sigma^2+Rt/3}\right) \text{ where } \sigma^2:=\left\| \sum_{k} \E X_k^2 \right\|
\end{align}
\end{theorem}

This particular formulation of the Theorem can be found in \cite{Tropp:TailBoundsRM}.\\

\begin{lemma} \label{lemma:proj1}
For n big enough, $\p\left(n-2\beta \log(n) \leq \lambda_{\max}(\Gamma^{\1})\right)<n^{-\epsilon}$ for some $\epsilon>0$.
\end{lemma}
  \begin{proof} Let $\epsilon>0$. Let $Q=\frac{\1\1^T}{n}$ be the projection matrix onto the $\1$ space. Then:
\begin{align}
\Gamma^\1 &= Q^T \Gamma Q\\
&= Q^T \left( \sum_{i<j, j \in S(i)}\left(2\A-1- \alpha_{ij}^+\right)\Delta^+_{ij}+\sum_{i<j, j \notin S(i)} \left(2\B-1 -\alpha_{ij}^-\right) \Delta^-_{ij} \right) Q \\
&= \sum_{i<j, j \notin S(i)}  -\frac{4}{n}\cdot \left(2\B-1 -\alpha_{ij}^-\right)Q
\end{align}
using the fact that $\Delta_{ij}^+Q=0_{n}$ and the fact that $Q^T \Delta_{ij}^- Q =-\frac{4}{n}Q$. \\

We have
\begin{align}
\lambda_{\max}\left(-\frac{4}{n}\left(2\B-1 -\alpha_{ij}^-\right)Q\right) \leq \frac{4}{n}\left(2-\frac{2\beta \log(n)}{n}\right) =:R
\end{align}
and
\begin{align} \sigma^2 &= \left\| \sum_{i<j, j \notin S(i)} \mathbb{E}\left[\left(-\frac{4}{n}\left(2\B-1-\alpha_{ij}^-\right)Q\right)^2\right] \right\| \\
&= \left\|\sum_{i<j, j \notin S(i)} \frac{16}{n^2} \cdot 4 \cdot \B\left(1-\B\right)Q \right\| \\
&= 16 \cdot \B\left(1-\B\right)
\end{align}

We then apply Theorem (\ref{theorem:MatrixBernstein})
\begin{align}
& \p\left(\lambda_{\max}(\Gamma^\1) \geq n - 2\beta \log(n)\right) \\
& \leq n\cdot \exp \left( -n^2 \cdot \frac{\left(1-\frac{2\beta\log(n)}{n}\right)^2/2}{16\B\left(1-\B\right)+\frac{4}{3}\left(2-2\B\right)\left(1-2\B\right)}\right)\\
& \leq n^{-\epsilon}
\end{align}
For big n this is clearly verified as $e^{-n^2}=o\left(n^{-(1+\epsilon)}\right)$.
\end{proof}

\begin{lemma} \label{lemma:projorth1}
If $(\alpha-\beta)^2>8(\alpha+\beta)+\frac{8}{3}(\alpha-\beta)$ then $\p\left(\lambda_{\max}(\Gamma^{\perp \1}) \geq (\alpha -\beta) \log(n)\right)< n^{-\epsilon}$ for some $\epsilon>0$.
\end{lemma}
\begin{proof} Let $P=I_{n}-\frac{\1 \1^T}{n}$ be the projection matrix onto the $\perp \1$ space. We have:
\begin{align}
\Gamma^{\perp \1} &= P^T \Gamma P\\
&=\Gamma_1+\Gamma_2^{\perp \1}\\
\text{where } \Gamma_1 &= \sum_{i<j, j \in S(i)}\left(2\A-1 - \alpha^{+}_{ij}\right) \cdot \Delta^+_{ij} \\
\Gamma_2^{\perp \1}&=\sum_{i<j, j \notin S(i)} \left( \frac{2\beta \log(n)}{n} - 1 -\alpha_{ij}^- \right) P^T \Delta_{ij}^- P 
\end{align}

$R=\max(R_1,R_2)$ where $R_1$ corresponds to $\Gamma_1$ and $R_2$ corresponds to $\Gamma_2^{\perp \1}$.
\begin{align}
\lambda_{\max}\left(\left(2\A-1 - \alpha^{+}_{ij}\right) \cdot \Delta^+_{ij} \right) &\leq \frac{4 \alpha \log(n)}{n} =:R_1 
\end{align}
\begin{align}
\lambda_{\max}\left(\left(2\B-1 - \alpha^{-}_{ij}\right) \cdot P^T \Delta^-_{ij}P  \right) & \leq   \lambda_{\max}\left(\left(2\B-1 - \alpha^{-}_{ij}\right) \cdot \Delta^-_{ij} \right) \leq 4 =: R_2
\end{align}
Hence we take $R:=4$.\\

$\sigma^2=\sigma_1^2+\sigma_2^2$ where $\sigma_1^2$ corresponds to $\Gamma_1$ and $\sigma_2^2$ corresponds to $\Gamma_2^{\perp \1}$.\\
\begin{align}
\sigma_1^2 &=\left\| \sum_{i<j, j \in S(i)}\E\left(2\A-1 - \alpha^{+}_{ij}\right)^2 \cdot (\Delta^+_{ij})^2  \right\| \\
&=\left\| \sum_{i<j, j \in S(i)}\frac{4\alpha \log(n)}{n}\left(1-\frac{\alpha \log(n)}{n}\right) \cdot 2\Delta^+_{ij} \right\| \\
&= \left\| \frac{4\alpha \log(n)}{n}\left(1-\frac{\alpha \log(n)}{n}\right) \cdot 2M  \right\| \\
&= 4\alpha\log(n)\left(1-\frac{\alpha \log(n)}{n} \right)
\end{align}
where
\begin{equation*}
 M=\begin{bmatrix} \frac{n}{2}-1 & \text{ } & -1 & \vrule & \text{ } & \text{ }  & \text{ } \\
\text{ } & \ddots & \text{ } & \vrule & \text{ }& 0 & \text{ } \\
-1 & \text{ } & \frac{n}{2}-1 & \vrule & \text{ } & \text{ } & \text{ } \\ \hline
 \text{ } & \text{ }  & \text{ } & \vrule & \frac{n}{2}-1 & \text{ } & -1  \\
 \text{ }& 0 & \text{ }  & \vrule & \text{ } & \ddots & \text{ } \\
\text{ } & \text{ } & \text{ }& \vrule & -1 & \text{ } & \frac{n}{2}-1\\
\end{bmatrix}\\
\end{equation*}

\begin{align}
\sigma_2^2 &= \left\| \sum_{i<j, j \notin S(i)}\E\left(2\B-1 - \alpha^{-}_{ij}\right)^2 \cdot \left(P^T \Delta^-_{ij} P\right)^2 \right\| \\
& \leq  \left\| \sum_{i<j, j \notin S(i)} 4\B\left(1-\B\right) \cdot P^T \left(-2\Delta^-_{ij}\right) P \right\| \\
&=  \left\| 4\B\left(1-\B\right) \cdot (-2M)\right\| \\
&= 4\beta \log(n) \left(1-\B \right)
\end{align}

We deduce
\begin{align}
\sigma^2= 4\alpha\log(n)\left(1-\frac{\alpha \log(n)}{n} \right) + 4\beta \log(n) \left(1-\B \right)
\end{align}
We then apply Theorem (\ref{theorem:MatrixBernstein}) using the values found previously and obtain
\begin{align}
&\p\left(\lambda_{\max}(\Gamma^{\perp \1}_2) \geq (\alpha-\beta)\log(n)\right)\\
& \leq n \cdot \exp \left(-\frac{(\alpha-\beta)^2\log(n)}{8\alpha\left(1-\frac{\alpha \log(n)}{n} \right) + 8\beta\left(1-\B \right) +\frac83(\alpha-\beta)} \right)\\
& \leq n^{-\epsilon}
\end{align}
This is equivalent to
\begin{align}
(\alpha-\beta)^2 > 8(\alpha+\beta) + \frac83(\alpha-\beta).
\end{align} 
\end{proof}

\subsection{Full Recovery Algorithm Proof}

\begin{lemma} \label{prop:degH1}
With high probability, the degree of any node in $H_1$ is at most $2 \frac{C}{\log(n)}n$.
\end{lemma}
\begin{proof}
Let $(Y_i)_{i=1..n}$ be a sequence of iid Bernouilli random variables of parameter $\frac{C}{\log(n)}$. Consider a node v in $H_1$. $H_1$ being an Erdos-Renyi graph on n vertices, we have that $\deg(v)= \sum_{i=1}^{n-1} Y_i$. Define $Y=\sum_{i=1}^n Y_i$. We have $Y \geq \deg(v)$ hence if $\p\left( Y \geq 2 \frac{C}{\log(n)}n \right) \rightarrow 0$ when $n \rightarrow \infty$, then $\p(\deg(v) \geq 2 \frac{C}{\log(n)}n) \rightarrow 0$ when $n \rightarrow \infty$ and we will have proved the result.\\

As $Y_i \in [0,1] \forall$ $i$ and $\E Y=\frac{C}{\log(n)}n$, using a Chernoff bound we get
\begin{align}
\p \left( Y \geq 2 \frac{C}{\log(n)}n \right) \leq \exp \left( - \frac{1}{4} \cdot \frac{C}{\log(n)}n \right)
\end{align}
where the right hand side goes to 0 when $n \rightarrow \infty$ as C is fixed. Hence using a union bound on all nodes
\begin{align}
\p\left( \exists \text{ a node s. t. its degree is more than } 2\frac{C}{\log(n)}n\right) &\leq n \p \left( Y \geq 2 \frac{C}{\log(n)}n \right) \\
&\leq n \cdot  \exp \left( - \frac{1}{4} \cdot \frac{C}{\log(n)}n \right) \rightarrow 0
\end{align}
when $n \rightarrow \infty$. 
\end{proof}

\end{document}